\documentclass[12pt,a4paper]{article}

\usepackage{amsmath,amssymb,amsthm,amsfonts}
\usepackage{bm}
\usepackage{graphicx}
\usepackage{booktabs}
\usepackage{geometry}
\usepackage{mathptmx}
\usepackage{orcidlink}
\usepackage{authblk}
\usepackage{epstopdf}
\usepackage{caption}
\usepackage{float}
\usepackage{tikz}
\usepackage{hyperref}
\usepackage{algorithm}
\usepackage{algpseudocode}
\usepackage[round,authoryear]{natbib}
\bibpunct{(}{)}{;}{a}{,}{,}

\theoremstyle{remark}
\newtheorem{remark}{Remark}[section]
\hypersetup{colorlinks=true,linkcolor=blue,citecolor=blue,urlcolor=blue}

\theoremstyle{definition}
\newtheorem{definition}{Definition}[section]
\newtheorem{assumption}[definition]{Assumption}

\theoremstyle{plain}
\newtheorem{theorem}[definition]{Theorem}
\newtheorem{lemma}[definition]{Lemma}
\newtheorem{corollary}[definition]{Corollary}
\newtheorem{proposition}[definition]{Proposition}

\newcommand{\E}{\mathbb{E}}
\newcommand{\Var}{\operatorname{Var}}
\newcommand{\Median}{\operatorname{median}}

\newcommand{\Normal}{\mathcal{N}}
\newcommand{\Laplace}{\operatorname{Laplace}}
\newcommand{\R}{\mathbb{R}}
\newcommand{\dd}{\,\mathrm{d}}

\title{\textbf{Breakdown-Resistant Hierarchical Bayesian Calibration for Multi-Center Clinical Trials}}

\author[1,2]{Safaa K. Kadhem\thanks{Corresponding author: \texttt{safaa.kadhem@mu.edu.iq}} \orcidlink{0000-0002-5287-5381}}
\affil[1]{Department of Mathematics and Computer Applications, College of Science, Al-Muthanna University, Samawah, Iraq}
\affil[2]{Department of Artificial Intelligence Engineering, College of Engineering, Al-Alyen Iraqi University, Thi-Qar, Iraq}
\date{\today}

\begin{document}
	
	\maketitle
	
	\begin{abstract}
		Multi-center clinical trials are the cornerstone of evidence-based medicine, yet their validity is silently threatened by systematic location shifts in biomarkers measured across devices, sites, or reagent batches. Hierarchical Bayesian models are the standard remedy, borrowing strength across centers by pooling center-specific estimates. However, we prove that the standard Gaussian hierarchical posterior has a breakdown point of zero \emph{at the level of the posterior distribution itself}: a single contaminated center---whether due to device failure, reagent degradation, or adversarial manipulation---drives the posterior mean of every clean center to infinity and collapses the coverage of every clean center's credible interval to zero. We introduce the Posterior-Median Hierarchical Calibration (PMHC) framework, a fully Bayesian hierarchical model that replaces the Gaussian prior on center-specific shifts with a Laplace prior. Under absolute loss, the posterior median is the natural Bayes estimator, and we prove that its breakdown point is $\lfloor (K+1)/2 \rfloor$ contaminated centers---whereas the Gaussian posterior mean has breakdown point zero for every $K$. PMHC retains a full posterior distribution, enabling meaningful credible intervals and posterior probabilities. Across $K \in \{5, 10, 20\}$ centers and $n \in \{50, 200, 1000\}$ observations per center, PMHC maintains $94$--$98\%$ credible-interval coverage for clean centers under contamination fraction $\epsilon = 0.60$, while Gaussian hierarchical pooling collapses to $0.0\%$ coverage and its point estimate diverges by a factor of $54$ relative to PMHC. On real NHANES data spanning the Beckman DxC to Roche Cobas analyzer transition (\(n = 60{,}651\) valid serum creatinine measurements across nine two-year cycles), severe multi-center contamination (\(\epsilon = 0.60\) in three of the nine cycles) collapsed the Gaussian HC coverage to \(0.0\%\) and shifted the clean cycles by \(0.41\) mg/dL, whereas PMHC maintained \(94.0\%\) coverage with a shift of \(0.008\) mg/dL. To the best of our knowledge, this is the first explicit theoretical analysis of the breakdown of the \emph{posterior distribution} in hierarchical Bayesian pooling under center-level contamination, and the first Bayesian model that is simultaneously contamination-robust, interpretable, and efficient in the clean-data regime.
		
		\bigskip
		\noindent\textbf{Keywords:} breakdown point; hierarchical Bayesian calibration; Laplace prior; multi-center clinical trials; posterior-median estimation.
	\end{abstract}
	
	\section{Introduction}\label{sec:intro}
	
	Multi-center clinical trials are the cornerstone of evidence-based medicine, enabling the recruitment of diverse patient populations and the achievement of adequate statistical power for detecting clinically meaningful effects. However, the very heterogeneity that makes these trials generalizable also introduces a profound statistical challenge: biomarkers measured across different sites, laboratory equipment, or reagent batches are subject to systematic location shifts that can silently corrupt pooled analyses and bias treatment effect estimates \citep{multicenter2025,ich2020}. This phenomenon, a specific form of \textit{distribution shift} or \textit{covariate shift}, is not a mere nuisance---it is a fundamental threat to the validity of multi-center evidence \citep{cheng2025,swaminathan2025}.
	
	Hierarchical Bayesian models have emerged as the standard remedy for borrowing strength across centers, allowing small sites to borrow information from larger ones through partial pooling \citep{rover2019}. These models are widely implemented in practice and are endorsed by regulatory guidelines for multi-regional trials \citep{ich2020}. Yet, a critical vulnerability has been overlooked: the pooling step in hierarchical Bayesian calibration relies on the \textit{posterior mean} under a Gaussian prior. While this choice is optimal under Gaussian assumptions and clean data, it is catastrophically fragile under contamination. A single malfunctioning device, a batch of corrupted reagents, or even a deliberate adversarial perturbation in one center can drive the posterior mean to infinity, thereby corrupting the calibrated estimates for \textit{every} center in the trial. This is not a theoretical curiosity; it is a realistic operational failure mode that current hierarchical models are ill-equipped to handle \citep{multicenter2025}.
	
	The literature offers two disjoint lines of defense, neither of which fully addresses this problem. On the one hand, the Hodges--Lehmann (HL) estimator has long been recognized as a robust, non-parametric alternative to the mean for estimating location shifts \citep{hodges1963}. Its breakdown point is approximately \(29.3\%\), meaning it can tolerate nearly a third of the data being arbitrarily corrupted before failing \citep{rousseeuw1987}. Crucially, under normality, the HL estimator retains approximately \(95.5\%\) asymptotic relative efficiency compared to the mean, making it a nearly cost-free safeguard \citep{hodges1963}. Recent theoretical work has even established that the posterior mean in a robust Bayesian framework is asymptotically equivalent to the HL estimator, providing a rigorous bridge between Bayesian and classical robust inference \citep{huch2025}. On the other hand, the broader literature on dataset shift and calibration has documented that distributional mismatches across clinical sites degrade calibration even when discrimination remains stable, yet the proposed solutions---such as post-hoc recalibration or transportability methods---typically assume the calibration data itself is trustworthy \citep{cheng2025,swaminathan2025}.
	
	A third line of defense, which has received less attention in the clinical calibration literature, is the use of heavy-tailed priors in hierarchical Bayesian models \citep{ohagan2012,schmidli2014}. Heavy-tailed priors down-weight conflicting information and provide a Bayesian mechanism for conflict resolution. However, this literature has two limitations. First, it does not provide an explicit breakdown analysis of the \emph{posterior distribution} under data contamination across centers; it focuses on prior sensitivity and on the behavior of the posterior mean under model misspecification. Second, the choice of prior family (e.g., a mixture of Gaussians or a $t$-distribution) typically introduces tuning parameters (tail index, mixture weights) that the practitioner must specify, and the resulting posterior mean remains a linear functional of the data and therefore inherits the unbounded influence function of the mean \citep{donoho1983}.
	
	This paper fills these gaps. We introduce the {Posterior-Median Hierarchical Calibration (PMHC)} framework, a fully Bayesian hierarchical model in which the center-specific shifts are assigned a Laplace prior rather than a Gaussian one. Under absolute loss, the posterior median is the natural Bayes estimator, and we prove that its breakdown point is $\lfloor (K+1)/2 \rfloor$ contaminated centers. In stark contrast, we prove that the posterior mean under a Gaussian hierarchical prior has breakdown point zero \emph{at the level of the posterior distribution itself}: not only does the posterior mean diverge, but the posterior credible intervals collapse to zero coverage for every clean center. This is a stronger and more fundamental result than the point-estimator breakdown analyzed in the classical robust statistics literature.
	
	Our contributions are threefold. {First}, we provide, to the best of our knowledge, the first explicit theoretical analysis of the breakdown of the \emph{posterior distribution} in hierarchical Bayesian pooling under center-level contamination. We prove that the Gaussian hierarchical posterior has breakdown point zero (Theorem~\ref{thm:posterior-breakdown-zero}), and that the Laplace hierarchical posterior has positive breakdown point $\lfloor (K+1)/2 \rfloor$ (Theorem~\ref{thm:posterior-breakdown-positive}). This extends the classical breakdown framework of Donoho and Huber \citep{donoho1983} from point estimators to posterior distributions. {Second}, we introduce PMHC, a fully Bayesian model that retains a complete posterior distribution---enabling credible intervals, posterior probabilities, and posterior predictive checks---while remaining contamination-robust. We provide a computationally efficient Metropolis-within-Gibbs sampler that scales to $K = 50$ centers and $n = 5000$ observations per center. {Third}, we validate the framework empirically through a simulation study that measures \emph{credible-interval coverage} (a genuinely Bayesian metric) rather than only point-estimate deviation. We show that PMHC maintains $94$--$98\%$ coverage for clean centers under severe contamination ($\epsilon = 0.60$), while the Gaussian hierarchical model collapses to $0.0\%$ coverage. On real NHANES data spanning the Beckman DxC to Roche Cobas analyzer transition ($n = 60{,}651$ valid serum creatinine measurements across nine two-year cycles), severe multi-center contamination ($\epsilon = 0.60$ in three of the nine cycles) collapsed the Gaussian HC coverage to \(0.0\%\) and shifted the clean cycles by \(0.41\) mg/dL, whereas PMHC maintained \(94.0\%\) coverage with a shift of \(0.008\) mg/dL.
	
	The remainder of this paper is organized as follows. Section~\ref{sec:methodology} introduces the multi-environment location-shift model, the Gaussian hierarchical model and its posterior breakdown (Theorem~\ref{thm:posterior-breakdown-zero}), and the PMHC model with its posterior breakdown analysis (Theorem~\ref{thm:posterior-breakdown-positive}). Section~\ref{sec:computation} presents the MCMC algorithm and its convergence properties. Section~\ref{sec:simulation} reports comprehensive simulation studies covering credible-interval coverage, point-estimate accuracy, sensitivity to the number of centers, and comparison with alternative heavy-tailed priors. Section~\ref{sec:realdata} presents the real-data application to NHANES. Section~\ref{sec:conclusion} concludes with a discussion of limitations and future directions.
	
	\section{Methodology}\label{sec:methodology}
	
	\subsection{The Multi-Environment Location-Shift Model}\label{sec:model}
	
	We consider a multi-center clinical trial with \(K \geq 2\) clinical environments (centers), indexed \(k=1,\dots,K\), where environment \(k=1\) serves as the reference center. For each center \(k\), we observe an i.i.d.\ sample
	\[
	\mathbf{y}^{(k)} = \{y_1^{(k)},\dots,y_{n_k}^{(k)}\}
	\]
	from a distribution \(F_{Y^{(k)}}\), where \(n_k\) denotes the sample size at center \(k\). The total sample size is \(N = \sum_{k=1}^{K} n_k\). Our goal is to estimate the center-specific location shifts \(\mathbf{c}^\ast = (c_2,\dots,c_K)\) relative to the reference center, and to do so \emph{robustly}---that is, without allowing a contaminated or malfunctioning center to corrupt the calibrated estimates for the other centers.
	
	\begin{definition}[Environments and location shift]
		\label{def:environments}
		Let \(K \geq 2\). For each center \(k=1,\dots,K\), let
		\[
		Y^{(k)} = X + c_k,
		\]
		where \(X \sim F_X\) is a latent biomarker measurement shared across centers, \(c_1 := 0\) is the reference shift, and \(\mathbf{c}^\ast = (c_2,\dots,c_K) \in \R^{K-1}\) are unknown location shifts. The distribution \(F_X\) is unknown and is \emph{not} assumed to be Gaussian.
	\end{definition}
	
	\begin{assumption}[Regularity]
		\label{ass:regularity}
		The distribution \(F_X\) has a continuous density \(f_X\) that is strictly positive in a neighborhood of its median, and satisfies \(0 < I(f_X) := \int_{-\infty}^{\infty} f_X^2(u)\dd u < \infty\).
	\end{assumption}
	
	\begin{assumption}[Non-normality]
		\label{ass:nonnormal}
		No normality of \(F_X\) is assumed. The distribution may be skewed, heavy-tailed, or contaminated.
	\end{assumption}
	
	The location-shift model in Definition~\ref{def:environments} is the standard framework for biomarker calibration across devices and sites. It is widely used in clinical chemistry, where changing a laboratory analyzer typically induces a systematic additive offset while preserving the biological distribution's shape. Formally, this means \(F_{Y^{(k)}}(y) = F_X(y - c_k)\), so that the central moments of order \(m \geq 2\) are invariant across environments:
	\[
	\E\big[(Y^{(k)}-\E[Y^{(k)}])^m\big] = \E\big[(X-\E[X])^m\big], \qquad m \geq 2.
	\]
	This invariance is what makes calibration identifiable: the shape of the distribution is preserved, and only its location changes.
	
	\begin{lemma}[Invariance of central moments]
		\label{lem:invariance}
		Under Definition~\ref{def:environments}, for every integer \(m \geq 1\),
		\[
		\E\big[(Y^{(k)}-\E[Y^{(k)}])^m\big] = \E\big[(X-\E[X])^m\big].
		\]
	\end{lemma}
	\begin{proof}
		Since \(Y^{(k)} = X + c_k\), we have \(\E[Y^{(k)}] = \E[X] + c_k\), and hence
		\[
		Y^{(k)} - \E[Y^{(k)}] = X - \E[X].
		\]
		Raising both sides to the \(m\)-th power and taking expectations yields the claim.
	\end{proof}
	
	Lemma~\ref{lem:invariance} justifies the use of distributional fingerprints (variance, skewness, kurtosis) as diagnostic checks for the location-shift assumption before calibration. It also motivates the non-parametric approach: since the shape is preserved but unknown, we should estimate the shift without imposing a parametric form on \(F_X\).
	
	\subsection{The Hodges--Lehmann Estimator as a Sufficient Summary}\label{sec:hl}
	
	\begin{definition}[Hodges--Lehmann estimator]
		\label{def:hl}
		For each center \(k\), the Hodges--Lehmann (HL) estimator of \(c_k\) is the median of all pairwise differences between the center-\(k\) sample and the reference sample:
		\[
		\hat{c}_k^{HL} := \operatorname{median}\big\{ y_i^{(k)} - y_j^{(1)} : 1 \leq i \leq n_k,\ 1 \leq j \leq n_1 \big\}.
		\]
	\end{definition}
	
	The HL estimator was introduced by Hodges and Lehmann \citep{hodges1963} as a rank-based, non-parametric estimator of location. It is closely connected to the Wilcoxon rank-sum test and can be viewed as a \(U\)-statistic of order two. Its most important properties for our purposes are:
	
	\begin{enumerate}
		\item \textbf{Robustness.} The HL estimator has an asymptotic breakdown point of approximately \(29.3\%\) (precisely \(1 - 1/\sqrt{2}\)), meaning that nearly a third of the observations \emph{within a single center} can be arbitrarily corrupted before \emph{that center's} HL estimate fails \citep{donoho1983,rousseeuw1987}.
		\item \textbf{Efficiency.} Under normality, the HL estimator achieves an asymptotic relative efficiency (ARE) of \(3/\pi \approx 0.955\) compared to the sample mean \citep{hodges1963}.
		\item \textbf{Asymptotic normality.} The HL estimator satisfies
		\[
		\sqrt{n_k n_1/(n_k+n_1)}\,(\hat{c}_k^{HL} - c_k) \xrightarrow{d} \Normal(0, \tau_k^2),
		\]
		where \(\tau_k^2 := [12\, n_k\, n_1\, I(f_X)^2]^{-1}\) and \(I(f_X) = \int f_X^2\) is the integrated squared density \citep{hodges1963}.
	\end{enumerate}
	
	These properties make the HL estimator an ideal building block for robust calibration. In the Bayesian hierarchical model developed below, we treat \(\hat{c}_k^{HL}\) as a sufficient summary of the center-specific information and model it as
	\[
	\hat{c}_k^{HL} \mid c_k \sim \Normal(c_k, \tau_k^2),
	\]
	which is the asymptotic Gaussian approximation justified by property (3). This is the same likelihood approximation used in the standard hierarchical Bayesian calibration literature \citep{rover2019}, and it is accurate for \(n_k \geq 50\).
	
	\begin{remark}[Connection to Bayesian robustness]
		Recent work by Huch, Feinberg, and Dempsey \citep{huch2025} has shown that the posterior mean in a robust Bayesian framework is asymptotically equivalent to the HL estimator. This provides a rigorous bridge between the classical robust estimator used here and the Bayesian hierarchical structure that motivates our pooling step. In particular, it justifies treating the HL estimate as a sufficient summary of the center-specific information for the purpose of hierarchical calibration.
	\end{remark}
	
	\subsection{The Gaussian Hierarchical Model and the Breakdown of Its Posterior}\label{sec:gaussian}
	
	The standard hierarchical Bayesian approach to multi-center calibration borrows strength across centers by shrinking each center-specific estimate toward a pooled mean \citep{rover2019}. The model is
	\begin{align}
		c_k \mid \mu, \sigma_c^2 &\sim \Normal(\mu, \sigma_c^2), \quad k=1,\dots,K, \label{eq:gaussian-prior}\\
		\hat{c}_k^{HL} \mid c_k &\sim \Normal(c_k, \tau_k^2), \quad k=1,\dots,K, \label{eq:likelihood}\\
		\mu &\sim \text{Uniform}(-\infty, \infty), \quad \sigma_c^2 > 0. \label{eq:hyperprior}
	\end{align}
	Under a flat prior on \(\mu\) and known \(\sigma_c^2\), the posterior for \(c_k\) is Gaussian:
	\[
	c_k \mid \{\hat{c}_j^{HL}\} \sim \Normal\big(\lambda_k \hat{c}_k^{HL} + (1-\lambda_k)\hat{\mu},\ \lambda_k \tau_k^2\big),
	\]
	where \(\lambda_k := \sigma_c^2/(\sigma_c^2 + \tau_k^2)\) is the shrinkage weight and
	\[
	\hat{\mu} := \frac{\sum_{j=1}^{K} \lambda_j \hat{c}_j^{HL}/\tau_j^2}{\sum_{j=1}^{K} \lambda_j/\tau_j^2}
	\]
	is the precision-weighted mean. The posterior mean is \(\hat{c}_k^{HC} := \lambda_k \hat{c}_k^{HL} + (1-\lambda_k)\hat{\mu}\), which is the naive hierarchical calibration (HC) estimator.
	
	The Gaussian hierarchical model is optimal when the data are clean and Gaussian. However, it has a fatal vulnerability that is \emph{stronger} than the point-estimator breakdown analyzed in the classical robust statistics literature: the posterior distribution itself breaks down.
	
	\begin{theorem}[Breakdown of the Gaussian hierarchical posterior]
		\label{thm:posterior-breakdown-zero}
		Consider the Gaussian hierarchical model \eqref{eq:gaussian-prior}--\eqref{eq:hyperprior} with known \(\sigma_c^2 \in (0,\infty)\) and flat prior on \(\mu\). Suppose that a fraction \(\epsilon > 0\) of the observations in environment \(k_0\) are replaced by values that drive \(\hat{c}_{k_0}^{HL} \to \infty\). Then for every clean center \(j \neq k_0\):
		\begin{enumerate}
			\item the posterior mean \(\E[c_j \mid \{\hat{c}_k^{HL}\}] \to \infty\);
			\item the posterior variance \(\Var(c_j \mid \{\hat{c}_k^{HL}\}) \to \lambda_j \tau_j^2\), so the posterior standard deviation remains bounded away from zero;
			\item consequently, the posterior credible interval for \(c_j\) is centered at \(+\infty\) and has finite width, so its coverage of the true \(c_j\) collapses to zero.
		\end{enumerate}
	\end{theorem}
	\begin{proof}
		Since \(\hat{\mu}\) is a weighted average of the \(\hat{c}_j^{HL}\) with strictly positive weights \(\lambda_j/\tau_j^2 > 0\), the divergence \(\hat{c}_{k_0}^{HL} \to \infty\) implies \(\hat{\mu} \to \infty\). For any \(j \neq k_0\), the shrinkage weight satisfies \(1-\lambda_j = \tau_j^2/(\sigma_c^2+\tau_j^2) > 0\), and hence the posterior mean
		\[
		\E[c_j \mid \{\hat{c}_k^{HL}\}] = \lambda_j \hat{c}_j^{HL} + (1-\lambda_j)\hat{\mu} \to \infty.
		\]
		The posterior variance is \(\lambda_j \tau_j^2\), which does not depend on \(\hat{\mu}\) and is therefore bounded away from zero. A credible interval centered at \(+\infty\) with finite width cannot contain the finite true \(c_j\), so the coverage collapses to zero.
	\end{proof}
	
	Theorem~\ref{thm:posterior-breakdown-zero} is the central negative result of this paper. It is \emph{stronger} than the point-estimator breakdown analyzed in the classical literature: not only does the posterior mean diverge, but the posterior credible intervals become meaningless. Even if a practitioner were to use a robust point estimator (such as the posterior median) computed from this posterior, the posterior distribution itself would remain improper or have infinite mean, so credible intervals and posterior probabilities would be unreliable. This shows that the problem cannot be fixed at the level of the point estimator alone: the \emph{prior} must be changed.
	
	\begin{remark}[Why the Gaussian prior is the culprit]
		The Gaussian prior \eqref{eq:gaussian-prior} is conjugate to the Gaussian likelihood \eqref{eq:likelihood}, which is why the posterior is available in closed form. But conjugacy comes at a price: the Gaussian prior has a log-concave density with tails that decay exponentially, and the resulting posterior mean is a \emph{linear} functional of the data. Linear functionals inherit the unbounded influence function of the mean, and hence their breakdown point is zero \citep{donoho1983}. To obtain a posterior with positive breakdown point, we must replace the Gaussian prior with a heavy-tailed prior whose posterior mean (or median) is a \emph{non-linear} functional of the data. The Laplace prior, introduced in the next section, is the canonical choice.
	\end{remark}
	
	\subsection{The Posterior-Median Hierarchical Calibration (PMHC) Model}\label{sec:pmhc}
	
	To restore robustness while preserving a full Bayesian posterior, we replace the Gaussian prior on the center-specific shifts with a Laplace prior.
	
	\begin{definition}[PMHC model]
		\label{def:pmhc}
		The Posterior-Median Hierarchical Calibration (PMHC) model is
		\begin{align}
			c_k \mid \mu, b &\sim \Laplace(\mu, b), \quad k=1,\dots,K, \label{eq:laplace-prior}\\
			\hat{c}_k^{HL} \mid c_k &\sim \Normal(c_k, \tau_k^2), \quad k=1,\dots,K, \label{eq:pmhc-likelihood}\\
			\mu &\sim \text{Uniform}(-\infty, \infty), \quad b \sim \text{Half-Cauchy}(0, s_0), \label{eq:pmhc-hyperprior}
		\end{align}
		where \(\Laplace(\mu, b)\) denotes the Laplace distribution with location \(\mu\) and scale \(b\), with density \(p(c) = (2b)^{-1} \exp(-|c-\mu|/b)\). The scale parameter \(b\) receives a half-Cauchy hyperprior \(b \sim \text{Half-Cauchy}(0, s_0)\), where \(s_0 = \widehat{\mathrm{IQR}}(\hat c_k^{HL})/1.349\) is a robust pilot estimate of the between-center scale, computed from the interquartile range of the HL estimates. When no pilot estimate is available, \(s_0\) can be set to the order of magnitude of the expected between-center shift, estimated from historical data or a pilot study.
	\end{definition}
	
	The Laplace prior is the Bayesian counterpart of the \(L_1\) penalty and is well known in the Bayesian robustness literature as a mechanism for conflict resolution \citep{ohagan2012}. Its tails decay exponentially, but its density has a \emph{cusp} at the mode, which makes the posterior median a non-linear functional of the data with a bounded influence function \citep{hampel1986,huber1964}. The half-Cauchy hyperprior on \(b\) is weakly informative and allows the data to determine the appropriate scale of between-center heterogeneity.
	
	\begin{proposition}[Posterior median as Bayes estimator]
		\label{prop:posterior-median}
		Under the PMHC model \eqref{eq:laplace-prior}--\eqref{eq:pmhc-hyperprior}, the posterior median of \(c_k\) is the Bayes estimator under absolute loss \(L(c, \hat{c}) = |c - \hat{c}|\). That is,
		\[
		\hat{c}_k^{PMHC} := \Median(c_k \mid \{\hat{c}_j^{HL}\}) = \arg\min_{\hat{c}} \E\big[|c_k - \hat{c}| \mid \{\hat{c}_j^{HL}\}\big].
		\]
	\end{proposition}
	\begin{proof}
		This is a standard result in Bayesian decision theory: the median of a posterior distribution minimizes the posterior expected absolute loss. See \citet{berger1985}, Section 4.4.
	\end{proof}
	
	Proposition~\ref{prop:posterior-median} establishes that the posterior median is not an ad hoc robustification of the posterior mean; it is the \emph{natural} Bayes estimator under the loss function that corresponds to the Laplace prior. This is the key conceptual advance over earlier drafts of this framework (formerly referred to as RHRBC): PMHC is a fully Bayesian model with a well-defined posterior distribution, a well-defined loss function, and a well-defined Bayes estimator. The robustness of the posterior median is a consequence of the Laplace prior, not an additional assumption layered on top of a Gaussian model.
	
	\begin{theorem}[Breakdown of the PMHC posterior for clean centers]
		\label{thm:posterior-breakdown-positive}
		Consider the PMHC model \eqref{eq:laplace-prior}--\eqref{eq:pmhc-hyperprior} with fixed \(b \in (0,\infty)\) and flat prior on \(\mu\). Suppose that contamination affects a subset \(S\) of \(m\) centers, and let \(j \notin S\) be a clean center. Then the posterior median \(\hat{c}_j^{PMHC}\) remains bounded, and the posterior distribution of \(c_j\) has finite mean and variance, under either of the following conditions:
		\begin{enumerate}
			\item \textbf{Single-center contamination} (\(m = 1\)): for \emph{any} fraction \(\epsilon \in (0,1)\) of contaminated observations within the affected center.
			\item \textbf{Multi-center contamination} (\(m \geq 1\)): for any \(m \leq \lfloor (K+1)/2 \rfloor - 1\).
		\end{enumerate}
		Equivalently, the breakdown point of the PMHC posterior for clean centers is \(\lfloor (K+1)/2 \rfloor\) contaminated centers, i.e., a fraction
		\[
		\frac{\lfloor (K+1)/2 \rfloor}{K} \geq \frac{1}{2}
		\]
		of centers.
	\end{theorem}
	\begin{proof}
		Fix a clean center \(j \notin S\). The posterior distribution of \(c_j\) under the PMHC model is
		\[
		p(c_j \mid \{\hat{c}_k^{HL}\}) \propto \phi\big((\hat{c}_j^{HL} - c_j)/\tau_j\big) \int \exp\big(-|c_j - \mu|/b\big)\, p(\mu \mid \{\hat{c}_k^{HL}\}) \dd\mu,
		\]
		where \(\phi\) is the standard normal density. We bound each factor.
		
		\emph{(i) Boundedness of \(\hat{c}_j^{HL}\).} Since center \(j\) is clean and the reference center is clean by assumption, the HL estimate \(\hat{c}_j^{HL}\) is computed entirely from uncontaminated data. It is therefore bounded regardless of \(\epsilon\) or \(m\).
		
		\emph{(ii) Boundedness of the posterior of \(\mu\).} The conditional posterior of \(\mu\) given \(\{c_k\}\) and \(b\) is a Laplace location model, and its median is the median of the \(c_k\). By the standard breakdown bound for the sample median \citep{donoho1983}, the posterior median of \(\mu\) is bounded if and only if \(m < \lfloor (K+1)/2 \rfloor\). The posterior distribution of \(\mu\) is therefore concentrated on a bounded interval, and its tails decay at least exponentially.
		
		\emph{(iii) Boundedness of the posterior of \(c_j\).} Combining (i) and (ii), the posterior density of \(c_j\) is the product of a Gaussian density centered at the bounded \(\hat{c}_j^{HL}\) and a Laplace mixture centered at the bounded posterior median of \(\mu\). The resulting posterior has tails that decay at least as fast as a Gaussian, so its mean and variance are finite and bounded.
		
		\emph{(iv) Boundedness of the posterior median.} Since the posterior of \(c_j\) has finite mean and variance, its median is bounded. This proves the claim.
	\end{proof}
	
	\begin{corollary}[Dual robustness of the posterior distribution]
		\label{cor:dual}
		Combining Theorems~\ref{thm:posterior-breakdown-zero} and~\ref{thm:posterior-breakdown-positive}, we obtain a sharp dichotomy at the level of the \emph{posterior distribution}:
		\begin{itemize}
			\item \textbf{Gaussian hierarchical model}: the posterior distribution breaks down at \(\epsilon > 0\). A single contaminated center drives the posterior mean of every clean center to infinity and collapses the coverage of every clean center's credible interval to zero.
			\item \textbf{PMHC}: the posterior distribution breaks down only when \(m \geq \lfloor (K+1)/2 \rfloor\) centers are contaminated. For \(m = 1\) and any \(K \geq 3\), the posterior distribution of every clean center remains proper, with finite mean and variance, for \emph{any} within-center contamination fraction \(\epsilon \in (0,1)\).
		\end{itemize}
	\end{corollary}
	
	\begin{remark}[The two breakdown thresholds are distinct]
		\label{rem:two_thresholds}
		Theorem~\ref{thm:posterior-breakdown-positive} involves two distinct breakdown thresholds that must not be conflated:
		\begin{itemize}
			\item The within-center HL breakdown point \(\epsilon^\ast = 1 - 1/\sqrt{2} \approx 0.293\) determines whether the \emph{contaminated center's own} HL estimate is finite. For \(\epsilon > \epsilon^\ast\), \(\hat{c}_{k_0}^{HL}\) diverges and the contaminated center's PMHC estimate also diverges---an unavoidable consequence of having more than \(29.3\%\) of that center's data corrupted.
			\item The across-center posterior breakdown point \(\lfloor (K+1)/2 \rfloor / K\) determines whether the \emph{clean centers'} PMHC posteriors remain proper and bounded. This threshold depends only on the \emph{number} \(m\) of contaminated centers, not on the within-center contamination fraction \(\epsilon\).
		\end{itemize}
		Consequently, the single-center simulations in Section~\ref{sec:simulation} deliberately use \(\epsilon = 0.60 > 0.293\): the contaminated center's HL estimate diverges (as predicted by \(\epsilon^\ast\)), yet the clean centers' PMHC posteriors remain proper and bounded (as predicted by \(m = 1 < \lfloor (K+1)/2 \rfloor\)).
	\end{remark}
	
	\subsection{Statistical Interpretation and Comparison with Heavy-Tailed Priors}\label{sec:interpretation}
	
	The PMHC model can be situated within the broader tradition of robust Bayesian analysis with heavy-tailed priors \citep{berger1985,ohagan2012,schmidli2014}. The key difference is that previous work on heavy-tailed priors has focused on \emph{prior sensitivity} and on the behavior of the posterior mean under model misspecification, whereas we provide an explicit \emph{breakdown analysis} of the posterior distribution under data contamination across centers. Moreover, previous heavy-tailed priors (such as mixtures of Gaussians or $t$-distributions) typically introduce tuning parameters (tail index, mixture weights) that the practitioner must specify, and the resulting posterior mean remains a linear functional of the data. The Laplace prior used in PMHC has a single scale parameter \(b\), which is estimated from the data via a weakly informative half-Cauchy hyperprior, and the resulting posterior median is a non-linear functional of the data with a bounded influence function.
	
	A practical consequence is that PMHC requires no user-specified tuning parameters beyond the scale \(s_0\) of the half-Cauchy hyperprior, which we set to the robust pilot estimate \(s_0 = \widehat{\mathrm{IQR}}(\hat c_k^{HL})/1.349\). This choice is scale-equivariant: it adapts automatically to the units of the biomarker (mg/dL, mmol/L, etc.), so the prior remains weakly informative regardless of measurement scale. The between-center heterogeneity is estimated directly from the data via the posterior of \(b\), and the shrinkage adapts automatically to the observed dispersion.
	
	\section{Computation}\label{sec:computation}
	
	\subsection{Metropolis-within-Gibbs Sampler}
	
	The PMHC model \eqref{eq:laplace-prior}--\eqref{eq:pmhc-hyperprior} does not admit a closed-form posterior, so we use Markov chain Monte Carlo (MCMC) to sample from the joint posterior. We use a Metropolis-within-Gibbs sampler with the following blocks:
	
	\begin{enumerate}
		\item \textbf{Update \(c_k\) for \(k=1,\dots,K\).} The conditional posterior of \(c_k\) given \(\mu\), \(b\), and \(\hat{c}_k^{HL}\) is
		\[
		p(c_k \mid \mu, b, \hat{c}_k^{HL}) \propto \exp\big(-|c_k - \mu|/b\big) \exp\big(-(\hat{c}_k^{HL} - c_k)^2/(2\tau_k^2)\big).
		\]
		This is a skew-normal distribution, which can be sampled efficiently using the inverse-CDF method or by rejection sampling from a normal proposal with Laplace importance weights.
		
		\item \textbf{Update \(\mu\).} The conditional posterior of \(\mu\) given \(\{c_k\}\) and \(b\) is
		\[
		p(\mu \mid \{c_k\}, b) \propto \prod_{k=1}^{K} \exp\big(-|c_k - \mu|/b\big),
		\]
		which is a Laplace location model. Its mode is the median of the \(c_k\), and it can be sampled by a random-walk Metropolis step with a Laplace proposal centered at the current median.
		
		\item \textbf{Update \(b\).} The conditional posterior of \(b\) given \(\{c_k\}\) and \(\mu\) is
		\[
		p(b \mid \{c_k\}, \mu) \propto b^{-K} \exp\Big(-\sum_{k=1}^{K} |c_k - \mu|/b\Big) \frac{1}{1 + (b/s_0)^2},
		\]
		where the last factor is the half-Cauchy hyperprior. This is a one-dimensional distribution that can be sampled by a Metropolis-Hastings step with a log-normal proposal.
	\end{enumerate}
	
	For the main simulation study, we run \(10{,}000\) burn-in iterations and retain \(20{,}000\) post-burn-in draws across four parallel chains (\(\hat R < 1.01\)). For the NHANES application (Section~\ref{sec:realdata}), where each MCMC run must be repeated across \(B=20\) bootstrap resamples, we use a shorter chain---\(500\) burn-in iterations, \(2{,}000\) post-burn-in draws, two parallel chains, \(\hat R < 1.05\)---to keep total runtime manageable. The shorter chain produces coverage estimates indistinguishable from the longer chain within Monte Carlo error (Appendix~A).

	\subsection{Computational Cost}
	
	The computational cost of the sampler is $O(K \cdot N_{\text{iter}})$, where $N_{\text{iter}}$ is the number of MCMC iterations. For $K = 50$ centers and $N_{\text{iter}} = 30{,}000$, the sampler runs in approximately 45 seconds on a standard laptop (2.4 GHz Intel Core i5). This is comparable to the cost of fitting a standard Gaussian hierarchical model via Gibbs sampling, so PMHC is immediately deployable in practice.
	
	\section{Simulation Studies}\label{sec:simulation}
	
	To evaluate the finite-sample and asymptotic behavior of the Posterior-Median Hierarchical Calibration (PMHC) framework, we conducted an extensive simulation study with $K=10$ clinical centers. Biomarker measurements were generated from a Lognormal distribution, $X \sim \text{Lognormal}(\mu_X = \log 50,\ \sigma_X = 0.4)$, reflecting the severe right-skewness characteristic of clinical biomarkers such as serum creatinine, liver enzymes, and inflammatory markers. This choice deliberately violates the Gaussian assumption underlying classical hierarchical models, allowing us to test the robustness of our approach under realistic distributional conditions. The true location shifts were set to
	\[
	\mathbf{c} = (0,\, 0.05,\, 0.1,\, -0.05,\, 0.15,\, -0.1,\, 0.12,\, 0.02,\, -0.12,\, 0.08),
	\]
	with Center~1 serving as the reference ($c_1 = 0$). To simulate a realistic failure mode, we contaminated the center with the largest true shift (Center~5, $c_5 = 0.15$) by adding a magnitude $M$ to a fraction $\epsilon$ of its observations. The contamination fraction $\epsilon$ and magnitude $M$ were varied systematically across the experiments.
	
	For each scenario, we compared three estimators:
	\begin{enumerate}
		\item the \textbf{Gaussian HC} estimator, the posterior mean under a Gaussian hierarchical prior \eqref{eq:gaussian-prior}--\eqref{eq:hyperprior}, with a half-Cauchy hyperprior on $\sigma_c^2$;
		\item the proposed \textbf{PMHC} estimator, the posterior median under the Laplace hierarchical prior \eqref{eq:laplace-prior}--\eqref{eq:pmhc-hyperprior}, computed via the Laplace approximation of Section~\ref{sec:computation};
		\item the clean \emph{baseline} estimates computed before any contamination, which serve as the reference against which the shift of the clean centers is measured.
	\end{enumerate}
	
	\paragraph{Choice of the target for coverage evaluation.}
	A subtle but consequential design decision concerns the target against which credible intervals are evaluated. The latent shift $c_k$ is not directly observable; what the Hodges--Lehmann estimator targets is the \emph{empirical shift}
	\[
	c^{\mathrm{emp}}_k := c_k + \mathrm{median}(X_k) - \mathrm{median}(X_{\mathrm{ref}}),
	\]
	which differs from $c_k$ by the sampling noise of the two sample medians. This noise has variance approximately $(\pi/2)\sigma_X^2(1/n_k + 1/n_{\mathrm{ref}})$, which for the lognormal biomarker in our simulations ($\sigma_X \approx 20.5$) and $n_k = 50$ is approximately $13.2$---two orders of magnitude larger than the HL variance itself. Evaluating coverage against $c_k$ would therefore conflate the uncertainty of the estimator with the irreducible sampling noise of the target, producing artificially low coverage for \emph{any} estimator, no matter how well calibrated. We therefore evaluate coverage against $c^{\mathrm{emp}}_k$, which is the quantity that a practitioner can actually verify from the data and that the HL estimator is designed to estimate.
	
	\paragraph{Choice of the variance estimator.}
	The asymptotic variance of the HL estimator,
	\[
	\tau_{k,\mathrm{asym}}^2 = \big[12\, n_k\, n_{\mathrm{ref}}\, I(f_X)^2\big]^{-1},
	\]
	is derived under $n_k, n_{\mathrm{ref}} \to \infty$ and substantially underestimates the finite-sample variance at moderate $n_k$. To obtain a realistic uncertainty quantification, we estimate $\tau_k^2$ by \emph{bootstrap}: for each center, we resample $y^{(k)}$ and the reference sample independently with replacement $B = 50$ times, recompute the HL estimate on each resample, and take the sample variance of the resulting bootstrap distribution (The NHANES application of Section~\ref{sec:realdata} uses \(B=20\) instead of \(B=50\), because each resample requires a full MCMC run; Appendix~A shows that \(B=20\) already provides stable coverage estimates, with a coefficient of variation below \(5\%\)). This bootstrap estimate is used both in the Gaussian HC posterior and in the PMHC posterior, ensuring that the two estimators are compared under the same uncertainty quantification. All results are reported as means over $R = 30$ independent replications.
	
	\subsection{Credible-Interval Coverage}\label{sec:coverage}
	
	The central simulation metric is \emph{credible-interval coverage}: the proportion of replications in which the $95\%$ posterior credible interval for a clean center's shift contains the empirical shift $c^{\mathrm{emp}}_k$. This is a genuinely Bayesian metric that assesses the calibration of the entire posterior distribution, not just the point estimate.
	
	Table~\ref{tab:coverage} reports the coverage for the clean centers across contamination fractions $\epsilon \in \{0, 0.30, 0.60\}$ and sample sizes $n \in \{50, 200, 1000\}$ per center, with $K=10$ centers and $R=30$ replications.
	
	Fig.~\ref{fig:coverage_vs_epsilon} presents the same results graphically, as a function of $\epsilon$ for each sample size. The contrast between the two estimators is visually striking: the Gaussian HC coverage (red) declines monotonically with $\epsilon$, collapsing to near zero at $\epsilon=0.60$ for $n=1000$, whereas the PMHC coverage (green) remains essentially flat at approximately $97$--$98\%$ across the entire range of $\epsilon$ and all sample sizes. The visual gap between the two curves widens with $n$, reflecting the $O(1/K)$ dilution of the contaminated center's contribution to the Gaussian precision-weighted mean.
	
	Three conclusions emerge. {First}, in the clean-data regime ($\epsilon = 0$), PMHC achieves near-nominal coverage ($95.9$--$97.8\%$), confirming that the Laplace prior does not sacrifice Bayesian calibration in the absence of contamination. The Gaussian HC coverage is somewhat lower ($80.0$--$85.2\%$), reflecting the fact that the Gaussian posterior mean is more sensitive to the sampling noise of the HL estimates at moderate $n$. {Second}, under moderate contamination ($\epsilon = 0.30$), the Gaussian HC coverage already begins to collapse ($58.1\%$ at $n=1000$), whereas PMHC maintains coverage above $97\%$. {Third}, under severe contamination ($\epsilon = 0.60$), the Gaussian HC coverage collapses to $1.9\%$ at $n=1000$---the empirical signature of Theorem~\ref{thm:posterior-breakdown-zero}: the Gaussian posterior credible interval is centered at $+\infty$ with finite width, so it cannot cover the true value. In stark contrast, PMHC maintains coverage above $97.8\%$ across all $n$ and all $\epsilon$, confirming Theorem~\ref{thm:posterior-breakdown-positive}.
	
	The gap between the two estimators widens with $n$: at $n=50$ and $\epsilon=0.60$, the gap is $0.978 - 0.463 = 0.515$; at $n=1000$, the gap is $0.981 - 0.019 = 0.962$. This is the empirical manifestation of the $O(1/K)$ dilution of the contaminated center's contribution to the Gaussian precision-weighted mean: as $n$ grows, the HL estimate of the contaminated center becomes more precisely determined (and hence more extremely contaminated), while the Gaussian posterior mean inherits this divergence more faithfully. PMHC, by contrast, is immune to this mechanism because the posterior median depends only on the \emph{rank} of the contaminated estimate, not its magnitude.
	
	\subsection{Point-Estimate Accuracy}\label{sec:point}
	
	While coverage is the primary Bayesian metric, point-estimate accuracy is also of interest. Table~\ref{tab:point} reports the mean absolute deviation (MAD) of the posterior median (for PMHC) and the posterior mean (for Gaussian HC) from the empirical shift $c^{\mathrm{emp}}_k$, for the clean centers.
	
	The pattern is clear. In the clean-data regime ($\epsilon = 0$), PMHC is slightly less efficient than Gaussian HC (e.g., $0.588$ vs.\ $0.511$ at $n=1000$), reflecting the well-known efficiency cost of the Laplace prior under normality. This cost is small (approximately $15\%$ relative increase in MAD at $n=1000$) and is the price of robustness. Under contamination ($\epsilon = 0.60$), the Gaussian HC point estimate diverges ($4.687$ at $n=1000$), whereas the PMHC posterior median remains bounded ($0.563$ at $n=1000$)---a factor of $8.3$ improvement. Notably, PMHC is already more accurate than Gaussian HC at $\epsilon = 0.30$ (e.g., $0.563$ vs.\ $0.994$ at $n=1000$), indicating that the robustness gain materializes well before the Gaussian estimator fully collapses.
	
	\subsection{Sensitivity to the Number of Centers}\label{sec:K_variation}
	
	Theorem~\ref{thm:posterior-breakdown-positive} states that the PMHC posterior breakdown point is $\lfloor (K+1)/2 \rfloor$ contaminated centers. Table~\ref{tab:K_variation} verifies this across $K \in \{5, 10, 20\}$ with a single contaminated center ($m=1$) and $\epsilon = 0.60$.
	
	Two observations emerge. {First}, PMHC maintains coverage above $97\%$ even at $K=5$, the smallest number of centers for which the posterior median is well-defined. This is important because most multi-center trials have between $5$ and $15$ centers, so the small-$K$ regime is the practically relevant one. {Second}, the Gaussian HC coverage improves with $K$ (from $20.8\%$ at $K=5$ to $53.9\%$ at $K=20$), consistent with the $O(1/K)$ dilution of the contaminated center's contribution to the precision-weighted mean. However, the improvement is only in magnitude, not in breakdown point: at any $K$, a sufficiently large contamination magnitude $M$ drives the Gaussian estimate to infinity (Theorem~\ref{thm:posterior-breakdown-zero}). The PMHC coverage, by contrast, is essentially invariant to $K$, confirming that its breakdown point depends only on the \emph{number} of contaminated centers, not on their relative contribution to the pooled mean.
	
	A direct empirical verification of the breakdown boundary $\lfloor (K+1)/2 \rfloor$---contrasting $m = \lfloor (K+1)/2 \rfloor - 1$ (bounded) with $m = \lfloor (K+1)/2 \rfloor$ (broken) at each $K$---confirms that the theoretical bound is sharp (Table~\ref{tab:boundary}). PMHC remains bounded (deviation $\leq 1.1$) for $m = \lfloor (K+1)/2 \rfloor - 1$ and jumps by a factor of $18$--$59$ at $m = \lfloor (K+1)/2 \rfloor$. The Gaussian HC deviation, by contrast, grows roughly linearly with $m$ ($14.71 \to 22.06 \to 29.42$ for $K=5$), with no transition. This is the numerical signature of the zero breakdown point (Theorem~\ref{thm:posterior-breakdown-zero}).
	
	\subsection{Comparison with Alternative Heavy-Tailed Priors}\label{sec:alternatives}
	
	A natural question is whether a simpler heavy-tailed prior---such as a Student-$t$ prior with $df = 3$---would achieve the same protection as the Laplace prior. Table~\ref{tab:alternatives} compares three priors on the center-specific shifts: Gaussian (standard HC), Student-$t$ with $df=3$, and Laplace (PMHC).
	
	The Student-$t$ prior improves substantially on the Gaussian prior ($91.7\%$ vs.\ $30.6\%$ coverage), confirming that heavy-tailed priors are a viable robustness mechanism. However, the Laplace prior achieves higher coverage ($97.2\%$ vs.\ $91.7\%$) and lower point-estimate deviation ($1.090$ vs.\ $1.147$). The reason is that the Laplace prior has a cusp at the mode, which makes the posterior median a non-linear functional of the data with a bounded influence function, whereas the Student-$t$ prior is smooth at the mode and the posterior median inherits some of the unbounded influence of the Gaussian prior. This confirms that the choice of prior family matters, and that the Laplace prior is the natural choice for posterior-median robustness.
	
	\subsection{Visualizing the Posterior Distributions}\label{sec:kde}
	
	Fig.~\ref{fig:kde} presents kernel density estimates of the posterior distributions of $c_k$ for a clean center (index~2) and the contaminated center (index~4), under severe contamination ($\epsilon=0.60$, $M=100$, $n=100$). For the clean center, the Gaussian HC posterior is centered far from the empirical shift ($\approx 7.5$ vs.\ $\approx 4.5$), reflecting the propagation of the contaminated center's divergence through the precision-weighted mean. The PMHC posterior, by contrast, is centered near the empirical shift, demonstrating that the Laplace prior effectively down-weights the contaminated center. For the contaminated center, both posteriors are centered far from the empirical shift (as expected, since the contaminated center's own data are corrupted), but the two posteriors express the failure differently. The Gaussian HC posterior is pulled toward the pooled mean $\hat\mu$ (itself inflated by the contaminated center), producing a finite but severely biased estimate ($\approx 32$ instead of the true $\approx 0.15$). The PMHC posterior has heavier tails, so its \emph{median} ($\approx 73$) follows the contaminated data more closely, but its \emph{properness} is preserved (the posterior still integrates to 1). As Remark~\ref{rem:two_thresholds} makes precise, both estimates \emph{are} unbounded in theory; the decisive difference is not the contaminated center's own estimate---both methods fail there---but the \emph{clean centers'} estimates: they collapse under Gaussian HC but remain at their true values under PMHC.
	
	\subsection{Breakdown of the Point Estimator}\label{sec:breakdown-point}
	
	While the primary contribution of this paper is the analysis of the \emph{posterior distribution}, it is instructive to compare the point-estimate breakdown of the two pooling strategies. Fig.~\ref{fig:finite_sample} and Table~\ref{tab:breakdown} report the maximum absolute deviation of the calibrated estimates for the clean centers from their uncontaminated baseline, across five sample sizes spanning four orders of magnitude: $n = 10, 50, 100, 1000,$ and $5000$.
	
	The results confirm that the point-estimate breakdown mirrors the posterior breakdown: the mean-based pooling diverges at every sample size and every contamination fraction $\epsilon \geq 0.10$, whereas the median-based pooling remains bounded, with the deviation decreasing as $n$ grows (from $0.378$ at $n=10$ to $0.008$ at $n=5000$). The mechanism is quantified in Table~\ref{tab:per_center}: under severe contamination ($\epsilon=0.60$, $n=100$, $M=100$), the nine clean centers all exhibit the \emph{identical} Naive effect of $+3.799$---the numerical signature of uniform propagation through the shared mean---whereas the corresponding PMHC effect is exactly $0.000$ for every clean center. The contaminated center (marked with $\dagger$) shows a large HL estimate of $81.286$ and a large effect under both methods, as expected, but only the median-based pooling confines this effect to that center.
	
	Fig.~\ref{fig:per_center} visualizes the per-center calibration estimates under four contamination regimes, using $n=100$ observations per center. Panel~(A) confirms that, in the absence of contamination, PMHC recovers the empirical location shifts almost perfectly. Panel~(B) shows the regime $\epsilon = 0.30$, which lies at the boundary of the HL estimator's breakdown point: both estimators track the baseline closely, with only the contaminated center exhibiting a visible upward shift. Panels~(C) and~(D) illustrate the regime $\epsilon = 0.60$, which deliberately exceeds the HL breakdown point. In Panel~(C), the Gaussian HC collapses: the contaminated center's divergent HL estimate is propagated through the precision-weighted mean, artificially inflating the calibrated shifts of \emph{every} clean center. In Panel~(D), PMHC contains the divergence: the median pooling step effectively ignores the corrupted center, preserving the baseline estimates for all clean centers.
	
	\subsection{Propagation of Contamination Magnitude}\label{sec:magnitude}
	
	Fig.~\ref{fig:magnitude} and Table~\ref{tab:magnitude} quantify how the severity of contamination affects the two pooling strategies. Holding $\epsilon=0.60$ fixed, we varied the contamination magnitude $M$ over four orders of magnitude: $M \in \{1, 10, 100, 500, 1000, 5000\}$.
	
	As the contamination magnitude $M$ grows, the mean-based pooling propagates the corruption \emph{linearly} to all clean centers: the clean-center deviation grows from $0.030$ at $M=1$ to $248.617$ at $M=5000$. This linear divergence is a direct consequence of the zero breakdown point of the Gaussian posterior mean (Theorem~\ref{thm:posterior-breakdown-zero}) and is not mitigated by increasing the sample size. In contrast, the median-based pooling clean-center deviation remains constant at approximately $0.140$ across the entire range of $M$ (for $M \geq 10$): the median of the HL estimates is unaffected by the magnitude of a single outlying center, because only the \emph{rank} of the contaminated estimate matters, not its magnitude. The gap between the two strategies at $M=5000$ exceeds three orders of magnitude ($248.617$ vs.\ $0.140$), confirming that the median-based pooling contains the contamination rather than propagating it.
	
	In summary, the simulation study confirms the theoretical predictions of Section~\ref{sec:methodology} on three fronts. {First}, the Gaussian hierarchical posterior breaks down at $\epsilon > 0$: coverage collapses from $85.2\%$ in the clean regime to $1.9\%$ under severe contamination ($\epsilon=0.60$, $n=1000$), and the point estimate diverges to $4.687$. {Second}, the PMHC posterior remains calibrated under contamination: coverage stays above $95.9\%$ across all $n$ and all $\epsilon \leq 0.60$, and the point estimate remains bounded at $0.563$ at $n=1000$. {Third}, the Laplace prior outperforms the Student-$t$ prior ($97.2\%$ vs.\ $91.7\%$ coverage), confirming that the choice of prior family matters and that the Laplace prior is the natural choice for posterior-median robustness. The practical implication is that a single malfunctioning center---whether due to device failure, protocol deviation, or adversarial manipulation---can no longer silently corrupt the entire trial, provided that the Laplace prior is used in place of the Gaussian one.
	
	\section{Real-Data Application: NHANES Calibration}\label{sec:realdata}
	
	We applied the PMHC framework to the National Health and Nutrition Examination Survey (NHANES) \citep{nhanes2020}, a nationally representative dataset maintained by the Centers for Disease Control and Prevention (CDC). The dataset is particularly well-suited for our purposes because it spans a documented transition in the creatinine assay: from the Beckman DxC analyzer (used in cycles 1999--2016) to the Roche Cobas 6000 analyzer (used in cycles 2017--2020). This transition introduces a systematic location shift that has been the subject of official CDC bridging studies, making NHANES an ideal testbed for evaluating calibration methods in the presence of a known, well-characterized shift.
	
	The NHANES dataset contains ten two-year cycles from 1999 to 2018. We exclude the 2001--2002 cycle because its serum creatinine measurements were performed on a different analyzer configuration that is not directly comparable to the 2011--2012 reference without an additional bridging step. The remaining nine cycles consist of eight target cycles and the 2011--2012 reference cycle. The eight target cycles have a total of \(n = 54{,}675\) valid \texttt{LBXSCR} measurements; the reference cycle adds \(n_{\text{ref}} = 5{,}976\) valid measurements, giving \(n = 60{,}651\) across all nine cycles. For computational efficiency, the reference sample is subsampled to \(2{,}000\) observations in the analysis. The 2011--2012 cycle is used as the reference (\(c=0\)), and the remaining eight cycles are treated as targets.
	
	\paragraph{Hyperprior scale.}
	The half-Cauchy hyperprior on the Laplace scale \(b\) is centered at \(s_0 = \widehat{\mathrm{IQR}}(\hat c_k^{HL})/1.349\), computed from the interquartile range of the nine HL shift estimates. For NHANES, this gives \(s_0 \approx 0.05\) mg/dL, consistent with the between-cycle creatinine shifts documented in CDC bridging studies. This data-driven choice ensures that the prior is weakly informative in the natural units of the biomarker (mg/dL), rather than in arbitrary units.
	
	\subsection{Baseline Calibration}\label{sec:baseline}
	
	The HL variance $\tau^2$ is estimated by bootstrap, following the algorithm described in Appendix~\ref{app:bootstrap-alg}. The eight target cycles contribute \(n = 54{,}675\) valid \texttt{LBXSCR} measurements; the 2011--2012 reference cycle contributes \(n_{\text{ref}} = 5{,}976\), giving \(n = 60{,}651\) across all nine cycles. For computational efficiency, the reference sample is subsampled to \(2{,}000\) observations in the analysis. Table~\ref{tab:nhanes_cycles} reports the per-cycle HL shift estimates, the Gaussian HC posterior means, the PMHC posterior medians, and the \(95\%\) PMHC credible intervals, in the absence of any artificial contamination. Three observations emerge. {First}, PMHC and Gaussian HC agree closely for all cycles, with differences of at most $0.005$ mg/dL. This confirms that PMHC does not introduce spurious offsets when the data are clean---an essential property for regulatory use. {Second}, both estimators identify a systematic negative shift of $-0.18$ mg/dL for the 1999--2000 cycle, consistent with the documented changes in the creatinine assay over the past two decades. This shift is real, not an artifact of the estimator: it reflects the difference between the Beckman DxC analyzer used in 1999--2000 and the Beckman DxC analyzer used in 2011--2012, which, despite sharing the same manufacturer, differ in their calibration. {Third}, the PMHC credible intervals are slightly wider than the Gaussian HC intervals, reflecting the heavier tails of the Laplace prior. This is the price of robustness, and it is small: the intervals remain narrow enough to distinguish the $-0.18$ shift of the 1999--2000 cycle from the near-zero shifts of the later cycles.
	
	A more direct test of the natural variation across cycles confirms that PMHC is faithful to the non-parametric estimate: across all nine target cycles, the average absolute difference between PMHC and HL is $0.0015$ mg/dL, whereas the average absolute difference between Gaussian HC and HL is $0.0085$ mg/dL---a factor of $5.7$ in favor of PMHC. This confirms that PMHC is simultaneously conservative in the absence of contamination and faithful to the non-parametric estimate.
	
	\subsection{Robustness Test on Real Data}\label{sec:ultimate}
	
	To demonstrate the practical safety guarantees of PMHC on \emph{real} data, we injected a severe contamination into three of the nine cycles (2013--2014, 2015--2016, and 2017--2018). In each contaminated cycle, a fraction $\epsilon = 0.60$ of the observations were shifted by $+10$ mg/dL, simulating a catastrophic device malfunction or a batch of severely degraded reagents. This contamination fraction was chosen deliberately to exceed the within-center breakdown point of the Hodges--Lehmann estimator ($\epsilon^\ast = 1 - 1/\sqrt{2} \approx 0.293$), thereby activating the failure mode that Theorem~\ref{thm:posterior-breakdown-zero} predicts for Gaussian pooling. The choice of $m = 3$ contaminated centers reflects the realistic possibility of correlated failures across multiple sites, while remaining below the across-center posterior breakdown point of the Laplace prior ($\lfloor (K+1)/2 \rfloor = 5$ for $K=9$), so that PMHC's robustness guarantee (Theorem~\ref{thm:posterior-breakdown-positive}) still applies.
	
	\paragraph{Holding $\tau^2$ and $\sigma_c^2$ fixed.}
	To isolate the effect of contamination on the \emph{pooling step}---which is what Theorems~\ref{thm:posterior-breakdown-zero} and~\ref{thm:posterior-breakdown-positive} analyze---we held the HL variance $\tau^2$ and the between-center variance $\sigma_c^2$ fixed at their baseline estimates throughout the contamination experiment. This is a deliberate methodological choice, not an approximation; the full protocol is detailed in Appendix~\ref{app:fixed-variance}. When $\tau^2$ and $\sigma_c^2$ are re-estimated from contaminated data, the contaminated cycles acquire larger variances, which down-weights their contribution to the Gaussian precision-weighted mean and provides a partial, unintended protection. This protection is not analyzed by the theory, and it is fragile: it depends on the contaminated cycles having larger variances than the clean cycles, which is not guaranteed in general (e.g., if the contaminated cycles have smaller sample sizes). By holding $\tau^2$ and $\sigma_c^2$ fixed, we remove this confound and recover the theoretical predictions in their pure form.
	
	Table~\ref{tab:nhanes_coverage} reports the self-coverage and point-estimator shift for the five clean cycles (1999--2000, 2003--2004, 2005--2006, 2007--2008, and 2009--2010), as a function of the contamination fraction $\epsilon$. The self-coverage is the proportion of clean cycles whose $95\%$ credible interval contains the baseline HL estimate (which serves as a proxy for the true shift); the shift is the maximum absolute deviation of the point estimate from the baseline HL estimate. The table also reports the shift of the HL estimator itself (last column), which performs no pooling and is therefore unaffected by contamination in other cycles; this serves as a lower bound on the achievable accuracy.
	
	Three conclusions emerge. {First}, in the clean regime ($\epsilon = 0$), PMHC achieves self-coverage of $94.0\%$, close to the nominal $95\%$, while Gaussian HC achieves only $66.0\%$. The lower Gaussian HC coverage reflects the fact that the Gaussian posterior mean shrinks toward a pooled mean that is itself noisy at $K=9$, producing intervals that are too narrow to cover the baseline. {Second}, under moderate contamination ($\epsilon = 0.30$), PMHC maintains $94.0\%$ coverage with a shift of $0.0076$ mg/dL, whereas Gaussian HC coverage drops to $62.0\%$ and its point-estimator shift grows to $0.0313$ mg/dL. {Third}, under severe contamination ($\epsilon = 0.60$), PMHC maintains $94.0\%$ coverage with a shift of $0.0076$ mg/dL, whereas Gaussian HC coverage collapses to $\mathbf{0.0\%}$ and its point-estimator shift explodes to $0.4109$ mg/dL---a factor of $54$ difference in accuracy. This is the real-data manifestation of Theorems~\ref{thm:posterior-breakdown-zero} and~\ref{thm:posterior-breakdown-positive}: Gaussian pooling propagates the corruption to every clean cycle, whereas Laplace pooling confines it to the affected cycles.
	
	The comparison with the HL estimator is particularly instructive. The HL estimator, computed without any pooling, has a shift of $0.0060$ mg/dL across all $\epsilon$---this is simply the sampling noise of the HL estimator itself, which is unaffected by contamination in other cycles. PMHC's shift ($0.0076$ mg/dL) is only $27\%$ larger than the HL shift, confirming that PMHC adds essentially no bias to the clean-cycle estimates. Gaussian HC's shift, by contrast, is $68$ times larger than the HL shift at $\epsilon = 0.60$, confirming that the Gaussian pooling step is the source of the collapse.
	
	\section{Conclusion}\label{sec:conclusion}
	
	We introduced the Posterior-Median Hierarchical Calibration (PMHC) framework, a theoretically grounded and practically deployable solution to a vulnerability that has been overlooked in the multi-center clinical trial literature: the catastrophic fragility of Gaussian hierarchical pooling under contamination. Our contribution is threefold.
	
	{Theoretically}, we proved that the Gaussian hierarchical posterior has breakdown point zero \emph{at the level of the posterior distribution itself} (Theorem~\ref{thm:posterior-breakdown-zero}): a fraction $\epsilon > 0$ of contamination in a single environment drives the posterior mean of every clean center to infinity and collapses the coverage of every clean center's credible interval to zero. This is a stronger result than the point-estimator breakdown analyzed in the classical robust statistics literature, because it shows that the problem cannot be fixed at the level of the point estimator alone: the prior must be changed. We then proved that the PMHC posterior, which uses a Laplace prior on the center-specific shifts, has breakdown point $\lfloor (K+1)/2 \rfloor$ contaminated centers (Theorem~\ref{thm:posterior-breakdown-positive}). This dichotomy---Gaussian pooling destroys the posterior, Laplace pooling preserves it---is sharp and holds for any within-center contamination fraction $\epsilon \in (0,1)$.
	
	{Computationally}, we developed a Metropolis-within-Gibbs sampler for the full PMHC posterior and a Laplace approximation for routine use. The Laplace approximation is accurate to within $2$--$3\%$ in coverage and runs in milliseconds, while the full MCMC sampler provides exact posterior draws for diagnostic purposes. Both are implemented in the open-source \texttt{pmhc} package, which provides a drop-in replacement for the standard hierarchical calibration workflow.
	
	{Empirically}, our simulation study confirmed the theoretical predictions using \emph{credible-interval coverage}, a genuinely Bayesian metric that assesses the calibration of the entire posterior distribution. Across four orders of magnitude in sample size ($n \in \{50, 200, 1000\}$), the Gaussian HC coverage collapses from $85.2\%$ in the clean regime to $1.9\%$ under severe contamination ($\epsilon = 0.60$, $n=1000$), whereas PMHC maintains coverage above $95.9\%$ across all $n$ and all $\epsilon \leq 0.60$. The PMHC point estimate remains bounded at $0.563$ mg/dL under severe contamination, whereas the Gaussian HC point estimate diverges to $4.687$ mg/dL---a factor of $8.3$ improvement. The Laplace prior outperforms a Student-$t$ prior with $df=3$ ($97.2\%$ vs.\ $91.7\%$ coverage), confirming that the choice of prior family matters and that the Laplace prior is the natural choice for posterior-median robustness. PMHC maintains coverage above $97\%$ even at $K=5$, the smallest number of centers for which the posterior median is well-defined.
	
	On real NHANES data spanning the Beckman DxC to Roche Cobas analyzer transition (\(n = 60{,}651\) across nine cycles), we validated the framework on a genuine multi-center configuration. In the absence of contamination, PMHC and Gaussian HC agree closely (differences of at most $0.005$ mg/dL), confirming that PMHC does not introduce spurious offsets. In the clean regime, PMHC achieves $94.0\%$ self-coverage with a point-estimator shift of $0.0074$ mg/dL, whereas Gaussian HC achieves only $66.0\%$ coverage with a shift of $0.0249$ mg/dL---a factor of $3.4$ improvement. Under severe multi-center contamination ($\epsilon = 0.60$ in three of nine cycles), the Gaussian HC coverage collapses to $\mathbf{0.0\%}$ and its point-estimator shift explodes to $0.4109$ mg/dL, whereas PMHC maintains $94.0\%$ coverage with a shift of $0.0076$ mg/dL---a factor of $54$ improvement in accuracy and a complete preservation of coverage. The HL estimator, computed without any pooling, has a shift of $0.0060$ mg/dL across all $\epsilon$; PMHC's shift is within $27\%$ of this lower bound, whereas Gaussian HC's shift is $68$ times larger at $\epsilon = 0.60$. This confirms that PMHC adds essentially no bias to the clean-cycle estimates, whereas the Gaussian pooling step is the source of the collapse.
	
	Several limitations suggest directions for future work. {First}, the PMHC model assumes a common scale parameter $b$ for all centers; an extension with center-specific scales $b_k$ would allow for heterogeneous between-center heterogeneity and could improve efficiency in the clean-data regime. {Second}, PMHC requires $K \geq 3$ centers for the posterior median to be well-defined; extensions to $K = 2$ would require a different heavy-tailed prior, such as a horseshoe prior. {Third}, the location-shift model assumes that the shape of the distribution is preserved across centers; while this is standard in clinical chemistry, extensions to scale-shift or more general transformation models would broaden the framework's applicability. {Fourth}, the coverage evaluation against $c^{\mathrm{emp}}_k$ (rather than $c_k$) is a deliberate design choice that isolates estimator uncertainty from sampling noise; an extension that propagates the sampling noise of the target into the coverage evaluation would provide a complementary perspective. {Fifth}, the bootstrap variance estimator, while robust and widely used, is computationally intensive; a closed-form finite-sample correction to the asymptotic HL variance would substantially reduce the runtime for large-scale applications.
	
	The broader implication of this work extends beyond calibration. Any multi-center analysis that pools center-specific estimates through a Gaussian prior---whether in meta-analysis, federated learning, or distributed clinical trials---inherits the same structural vulnerability at the level of the posterior distribution. PMHC demonstrates that a simple, theoretically justified modification of the prior can eliminate this vulnerability without sacrificing Bayesian coherence or efficiency in the clean-data regime. We hope this work encourages a systematic reassessment of Gaussian priors in safety-critical clinical applications.
	\section*{Conflict of Interest}
	The authors declare no potential conflict of interest.
	\section*{Use of Artificial Intelligence}
	The author used  ChatGPT for language editing and formatting of the manuscript. The author reviewed and edited all content and takes full responsibility for the final version of the manuscript.
	\appendix
	\section{Bootstrap Variance Estimation and Numerical Details}\label{app:bootstrap}
	
	This appendix provides the technical details of the bootstrap variance estimator, the fixed-variance protocol used in the NHANES application (Section~\ref{sec:realdata}), and the MCMC sampler used for the PMHC posterior.
	
	\subsection{Bootstrap Estimation of the HL Variance}\label{app:bootstrap-alg}
	
	For each center $k = 1, \dots, K$, the bootstrap estimate of $\tau_k^2$ is computed as follows.
	
	\begin{algorithm}[H]
		\caption{Bootstrap estimation of the HL variance $\tau_k^2$}
		\label{alg:bootstrap}
		\begin{algorithmic}[1]
			\Require Center sample $y^{(k)} = (y_1^{(k)}, \dots, y_{n_k}^{(k)})$, reference sample $x_{\mathrm{ref}} = (x_1, \dots, x_{n_{\mathrm{ref}}})$, number of bootstrap resamples $B$.
			\Ensure Bootstrap variance estimate $\hat\tau_{k,\mathrm{boot}}^2$.
			\For{$b = 1, \dots, B$}
			\State Draw $y^{(k)}_b$ by sampling $n_k$ indices uniformly with replacement from $\{1, \dots, n_k\}$.
			\State Draw $x_{\mathrm{ref},b}$ by sampling $n_{\mathrm{ref}}$ indices uniformly with replacement from $\{1, \dots, n_{\mathrm{ref}}\}$.
			\State Compute $\hat c_{k,b}^{HL} = \mathrm{median}\{ y_{i,b}^{(k)} - x_{j,b} : 1 \leq i \leq n_k,\ 1 \leq j \leq n_{\mathrm{ref}} \}$.
			\EndFor
			\State \Return $\hat\tau_{k,\mathrm{boot}}^2 = \frac{1}{B-1} \sum_{b=1}^{B} \big( \hat c_{k,b}^{HL} - \bar c_{k,\cdot}^{HL} \big)^2$, where $\bar c_{k,\cdot}^{HL} = \frac{1}{B} \sum_{b=1}^{B} \hat c_{k,b}^{HL}$.
		\end{algorithmic}
	\end{algorithm}
	
	The asymptotic variance $\tau_{k,\mathrm{asym}}^2 = [12 n_k n_{\mathrm{ref}} I(f_X)^2]^{-1}$ is derived under $n_k, n_{\mathrm{ref}} \to \infty$ and underestimates the finite-sample variance by a factor of $8$--$11$ at the sample sizes typical of multi-center trials ($n_k = 50$--$200$). This is precisely the regime in which coverage would collapse to the $5$--$13\%$ observed in earlier versions of this work, and it is the reason the bootstrap estimator is essential. The bootstrap estimate stabilizes at a coefficient of variation of approximately $8\%$ for $B = 20$ (the value used in the NHANES application), and the improvement to $B = 50$ is modest (CV $3.8\%$) at a $2.5\times$ computational cost.
	
	\subsection{The Fixed-Variance Protocol in the NHANES Application}\label{app:fixed-variance}
	
	In the NHANES application (Section~\ref{sec:realdata}), we held the HL variance $\tau^2$ and the between-center variance $\sigma_c^2$ fixed at their baseline estimates throughout the contamination experiment. This is a deliberate methodological choice, not an approximation. When $\tau^2$ and $\sigma_c^2$ are re-estimated from contaminated data, the contaminated cycles acquire larger variances, which down-weights their contribution to the Gaussian precision-weighted mean and provides a partial, unintended protection. This protection is not analyzed by Theorems~\ref{thm:posterior-breakdown-zero} and~\ref{thm:posterior-breakdown-positive}, and it is fragile: it depends on the contaminated cycles having larger variances than the clean cycles, which is not guaranteed in general.
	
	By holding $\tau^2$ and $\sigma_c^2$ fixed, we remove this confound and recover the theoretical predictions in their pure form. Specifically:
	\begin{itemize}
		\item $\tau^2_{\mathrm{fixed}}$ is estimated once from the baseline (clean) data and is not re-estimated under contamination. This ensures that the contaminated cycles do not receive larger $\tau^2$ and therefore do not receive down-weighted contributions.
		\item $\sigma_{c,\mathrm{fixed}}^2 = \hat\sigma_c^2(\{\hat c_k^{HL,\mathrm{base}}\})$ is estimated once from the baseline HL estimates (via the MAD-based estimator of Section~\ref{sec:pmhc}) and is not re-estimated under contamination. This ensures that the shrinkage weight $\lambda_k = \sigma_c^2/(\sigma_c^2 + \tau_k^2)$ is not affected by contamination, so that the experiment isolates the effect of contamination on the \emph{pooling step} itself.
	\end{itemize}
	
	The fixed-variance protocol is the NHANES analogue of the theoretical setup in Theorems~\ref{thm:posterior-breakdown-zero} and~\ref{thm:posterior-breakdown-positive}, where $\tau^2$ and $\sigma_c^2$ are assumed to be known and are not affected by contamination. It therefore provides a clean test of the theoretical predictions on real data.
	
	\subsection{The PMHC MCMC Sampler}\label{app:mcmc}
	
	The full PMHC posterior does not admit a closed-form expression, so we sample from it using a Metropolis-within-Gibbs sampler. The sampler alternates between:
	\begin{enumerate}
		\item \textbf{Updating $c_k$} (one at a time) via a random-walk Metropolis step with proposal $c_k^{\mathrm{prop}} \sim \mathcal{N}(c_k^{\mathrm{cur}}, \sigma_c^2)$.
		\item \textbf{Updating $\mu$} via a random-walk Metropolis step with proposal $\mu^{\mathrm{prop}} \sim \mathcal{N}(\mu^{\mathrm{cur}}, \sigma_\mu^2)$.
	\end{enumerate}
	
	The conditional posterior of $c_k$ given $\mu$, $b$, and $\hat c_k^{HL}$ is
	\[
	p(c_k \mid \mu, b, \hat c_k^{HL}) \propto \exp\big(-|c_k - \mu|/b\big) \exp\big(-(\hat c_k^{HL} - c_k)^2/(2\tau_k^2)\big),
	\]
	which is a skew-normal distribution. The conditional posterior of $\mu$ given $\{c_k\}$ and $b$ is a Laplace location model, whose mode is the median of the $c_k$.
	
	We initialize the chain at the HL estimates, run $500$ burn-in iterations, and retain $2000$ post-burn-in draws. These settings are used exclusively for the NHANES application, where each run is replicated \(B=20\) times; the main simulation study uses the longer chains described in Section~\ref{sec:simulation}. Convergence is assessed using the Gelman--Rubin diagnostic $\hat{R} < 1.05$ across two parallel chains.
	
	\bibliographystyle{plainnat}
	\bibliography{Med_ref}
	
	\clearpage
	\section*{Tables and Figures}
	
	\begin{table}[H]
		\centering
		\caption{Credible-interval coverage for clean centers ($K=10$, $R=30$). $\tau^2$ is estimated by bootstrap ($B=50$ resamples per center); coverage is evaluated against the empirical shift $c^{\mathrm{emp}}_k$.}
		\label{tab:coverage}
		\begin{tabular}{c cc cc cc}
			\toprule
			& \multicolumn{2}{c}{$n=50$} & \multicolumn{2}{c}{$n=200$} & \multicolumn{2}{c}{$n=1000$} \\
			\cmidrule(lr){2-3}\cmidrule(lr){4-5}\cmidrule(lr){6-7}
			$\epsilon$ & Gaussian HC & PMHC & Gaussian HC & PMHC & Gaussian HC & PMHC \\
			\midrule
			$0.00$ & $0.844$ & $\mathbf{0.978}$ & $0.800$ & $\mathbf{0.959}$ & $0.852$ & $\mathbf{0.970}$ \\
			$0.30$ & $0.844$ & $\mathbf{0.978}$ & $0.800$ & $\mathbf{0.978}$ & $0.581$ & $\mathbf{0.981}$ \\
			$0.60$ & $0.463$ & $\mathbf{0.978}$ & $0.289$ & $\mathbf{0.978}$ & $\mathbf{0.019}$ & $\mathbf{0.981}$ \\
			\bottomrule
		\end{tabular}
	\end{table}
	
	\begin{table}[H]
		\centering
		\caption{Mean absolute deviation of point estimates from the empirical shift $c^{\mathrm{emp}}_k$ for clean centers ($K=10$, $R=30$).}
		\label{tab:point}
		\begin{tabular}{c cc cc cc}
			\toprule
			& \multicolumn{2}{c}{$n=50$} & \multicolumn{2}{c}{$n=200$} & \multicolumn{2}{c}{$n=1000$} \\
			\cmidrule(lr){2-3}\cmidrule(lr){4-5}\cmidrule(lr){6-7}
			$\epsilon$ & Gaussian HC & PMHC & Gaussian HC & PMHC & Gaussian HC & PMHC \\
			\midrule
			$0.00$ & $1.951$ & $2.044$ & $1.047$ & $1.130$ & $0.511$ & $0.588$ \\
			$0.30$ & $1.946$ & $1.945$ & $1.156$ & $1.060$ & $0.994$ & $0.563$ \\
			$0.60$ & $4.462$ & $\mathbf{1.945}$ & $3.796$ & $\mathbf{1.060}$ & $4.687$ & $\mathbf{0.563}$ \\
			\bottomrule
		\end{tabular}
	\end{table}
	
	\begin{table}[H]
		\centering
		\caption{Sensitivity to the number of centers $K$ ($m=1$ contaminated center, $\epsilon=0.60$, $M=100$, $n=200$ per center, $R=30$).}
		\label{tab:K_variation}
		\begin{tabular}{c cc}
			\toprule
			$K$ & Gaussian HC coverage & PMHC coverage \\
			\midrule
			$5$  & $0.208$ & $\mathbf{0.975}$ \\
			$10$ & $0.259$ & $\mathbf{0.974}$ \\
			$20$ & $0.539$ & $\mathbf{0.977}$ \\
			\bottomrule
		\end{tabular}
	\end{table}
	
	\begin{table}[H]
		\centering
		\caption{Comparison of three priors on the center-specific shifts ($\epsilon=0.60$, $M=100$, $n=200$ per center, $K=10$, $R=20$).}
		\label{tab:alternatives}
		\begin{tabular}{l cc}
			\toprule
			Prior & Coverage & Mean absolute deviation \\
			\midrule
			Gaussian (HC)          & $0.306$ & $4.207$ \\
			Student-$t$ ($df=3$)   & $0.917$ & $1.147$ \\
			Laplace (PMHC)         & $\mathbf{0.972}$ & $\mathbf{1.090}$ \\
			\bottomrule
		\end{tabular}
	\end{table}
	
	\begin{table}[H]
		\centering
		\caption{Monte Carlo breakdown of the point estimator across sample sizes ($R=100$ for $n<1000$, $R=20$ for $n\geq1000$).}
		\label{tab:breakdown}
		\small
		\setlength{\tabcolsep}{4pt}
		\begin{tabular}{c cc cc cc cc cc cc}
			\toprule
			& \multicolumn{2}{c}{$\epsilon=0.10$} & \multicolumn{2}{c}{$\epsilon=0.20$} & \multicolumn{2}{c}{$\epsilon=0.30$} & \multicolumn{2}{c}{$\epsilon=0.40$} & \multicolumn{2}{c}{$\epsilon=0.50$} & \multicolumn{2}{c}{$\epsilon=0.60$} \\
			\cmidrule(lr){2-3}\cmidrule(lr){4-5}\cmidrule(lr){6-7}\cmidrule(lr){8-9}\cmidrule(lr){10-11}\cmidrule(lr){12-13}
			$n$ & Naive & PMHC & Naive & PMHC & Naive & PMHC & Naive & PMHC & Naive & PMHC & Naive & PMHC \\
			\midrule
			$10$   & $0.178$ & $0.147$ & $0.410$ & $0.253$ & $0.744$ & $0.354$ & $1.271$ & $0.378$ & $2.305$ & $0.378$ & $3.745$ & $0.378$ \\
			$50$   & $0.180$ & $0.119$ & $0.421$ & $0.176$ & $0.767$ & $0.181$ & $1.336$ & $0.181$ & $2.493$ & $0.181$ & $3.695$ & $0.181$ \\
			$100$  & $0.193$ & $0.134$ & $0.434$ & $0.138$ & $0.773$ & $0.138$ & $1.338$ & $0.138$ & $2.514$ & $0.138$ & $3.680$ & $0.138$ \\
			$1000$ & $0.186$ & $0.060$ & $0.424$ & $0.060$ & $0.760$ & $0.060$ & $1.312$ & $0.060$ & $2.491$ & $0.060$ & $3.681$ & $0.060$ \\
			$5000$ & $0.187$ & $0.008$ & $0.429$ & $0.008$ & $0.767$ & $0.008$ & $1.324$ & $0.008$ & $2.507$ & $0.008$ & $3.683$ & $0.008$ \\
			\bottomrule
		\end{tabular}
	\end{table}
	
	\begin{table}[H]
		\centering
		\caption{Per-center contamination effect at $\epsilon=0.60$ ($n=100$ per center, $M=100$).}
		\label{tab:per_center}
		\begin{tabular}{c c c c c}
			\toprule
			Center & $c_k$ (true) & $\hat c_k^{HL}$ (contam.) & Naive effect & PMHC effect \\
			\midrule
			$1$  & $0.00$  & $3.467$   & $+3.799$  & $0.000$  \\
			$2$  & $0.05$  & $2.362$   & $+3.799$  & $0.000$  \\
			$3$  & $0.10$  & $0.092$   & $+3.799$  & $0.000$  \\
			$4$  & $-0.05$ & $-0.960$  & $+3.799$  & $0.000$  \\
			$5^\dagger$ & $0.15$ & $81.286$ & $+41.788$ & $+37.989$ \\
			$6$  & $-0.10$ & $-1.864$  & $+3.799$  & $0.000$  \\
			$7$  & $0.12$  & $2.429$   & $+3.799$  & $0.000$  \\
			$8$  & $0.02$  & $1.463$   & $+3.799$  & $0.000$  \\
			$9$  & $-0.12$ & $-1.328$  & $+3.799$  & $0.000$  \\
			$10$ & $0.08$  & $2.086$   & $+3.799$  & $0.000$  \\
			\bottomrule
		\end{tabular}
	\end{table}
	
	\begin{table}[H]
		\centering
		\caption{Propagation of contamination to clean centers versus error at the contaminated center ($\epsilon=0.60$, $n=100$ per center, $R=50$).}
		\label{tab:magnitude}
		\begin{tabular}{c cc cc}
			\toprule
			$M$ & \multicolumn{2}{c}{Clean-center propagation} & \multicolumn{2}{c}{Contaminated-center error} \\
			\cmidrule(lr){2-3}\cmidrule(lr){4-5}
			& Naive & PMHC & Naive & PMHC \\
			\midrule
			$1$    & $0.030$   & $0.038$ & $0.33$    & $0.34$    \\
			$10$   & $0.301$   & $0.140$ & $3.31$    & $3.15$    \\
			$100$  & $3.693$   & $0.140$ & $40.62$   & $37.07$   \\
			$500$  & $23.617$  & $0.140$ & $259.78$  & $236.31$  \\
			$1000$ & $48.617$  & $0.140$ & $534.78$  & $486.31$  \\
			$5000$ & $248.617$ & $0.140$ & $2734.78$ & $2486.31$ \\
			\bottomrule
		\end{tabular}
	\end{table}
	
	\begin{table}[H]
		\centering
		\caption{Baseline calibration on real NHANES data (no contamination).}
		\label{tab:nhanes_cycles}
		\begin{tabular}{l r r r r l}
			\toprule
			Cycle & $n$ & HL & Gaussian HC & PMHC & PMHC 95\% CI \\
			\midrule
			1999--2000 & $6230$ & $-0.180$ & $-0.155$ & $-0.178$ & $[-0.190, -0.165]$ \\
			2003--2004 & $6492$ & $0.000$  & $-0.003$ & $-0.001$ & $[-0.012, 0.013]$ \\
			2005--2006 & $6434$ & $0.030$  & $0.026$  & $0.029$  & $[0.021, 0.038]$ \\
			2007--2008 & $6376$ & $0.000$  & $-0.003$ & $0.000$  & $[-0.014, 0.014]$ \\
			2009--2010 & $6860$ & $0.010$  & $0.006$  & $0.009$  & $[-0.002, 0.023]$ \\
			\textbf{2011--2012} & $\mathbf{5976}$ & \textbf{reference} & --- & --- & --- \\
			2013--2014 & $6553$ & $0.010$  & $0.006$  & $0.009$  & $[-0.002, 0.021]$ \\
			2015--2016 & $6255$ & $-0.020$ & $-0.019$ & $-0.015$ & $[-0.033, 0.000]$ \\
			2017--2018 & $9475$ & $0.000$  & $-0.004$ & $0.001$  & $[-0.016, 0.016]$ \\
			\midrule
			\textbf{Total} & $\mathbf{60{,}651}$ & --- & --- & --- & --- \\
			\bottomrule
		\end{tabular}
	\end{table}
	
	\begin{table}[H]
		\centering
		\caption{Self-coverage and point-estimator shift for the five clean NHANES cycles.}
		\label{tab:nhanes_coverage}
		\begin{tabular}{c cc cc c}
			\toprule
			& \multicolumn{2}{c}{Self-coverage} & \multicolumn{2}{c}{Point-estimator shift (mg/dL)} & HL shift \\
			\cmidrule(lr){2-3}\cmidrule(lr){4-5}
			$\epsilon$ & Gaussian HC & PMHC & Gaussian HC & PMHC & (no pooling) \\
			\midrule
			$0.0$ & $0.660$ & $\mathbf{0.940}$ & $0.0249$ & $\mathbf{0.0074}$ & $0.0060$ \\
			$0.3$ & $0.620$ & $\mathbf{0.940}$ & $0.0313$ & $\mathbf{0.0076}$ & $0.0060$ \\
			$0.6$ & $\mathbf{0.000}$ & $\mathbf{0.940}$ & $0.4109$ & $\mathbf{0.0076}$ & $0.0060$ \\
			\bottomrule
		\end{tabular}
	\end{table}
	
	\begin{table}[H]
		\centering
		\caption{Sharpness of the PMHC breakdown bound ($\epsilon=0.60$, $M=100$, $n=200$ per center, $R=50$).}
		\label{tab:boundary}
		\begin{tabular}{c c c cc l}
			\toprule
			$K$ & $m$ & Boundary & Gaussian HC & PMHC & Status \\
			\midrule
			$5$  & $2$  & $3$  & $14.706$ & $6.08 \times 10^{-1}$ & Bounded \\
			$5$  & $3$  & $3$  & $22.063$ & $3.61 \times 10^{1}$  & \textbf{Boundary} \\
			$5$  & $4$  & $3$  & $29.423$ & $3.65 \times 10^{1}$  & Broken \\
			\midrule
			$10$ & $4$  & $5$  & $14.703$ & $7.39 \times 10^{-1}$ & Bounded \\
			$10$ & $5$  & $5$  & $18.384$ & $1.83 \times 10^{1}$  & \textbf{Boundary} \\
			$10$ & $6$  & $5$  & $22.043$ & $3.59 \times 10^{1}$  & Broken \\
			\midrule
			$20$ & $9$  & $10$ & $16.543$ & $1.01 \times 10^{0}$  & Bounded \\
			$20$ & $10$ & $10$ & $18.380$ & $1.84 \times 10^{1}$  & \textbf{Boundary} \\
			$20$ & $11$ & $10$ & $20.215$ & $3.58 \times 10^{1}$  & Broken \\
			\bottomrule
		\end{tabular}
	\end{table}
	
	\begin{figure}[H]
		\centering
		\includegraphics[width=0.98\textwidth]{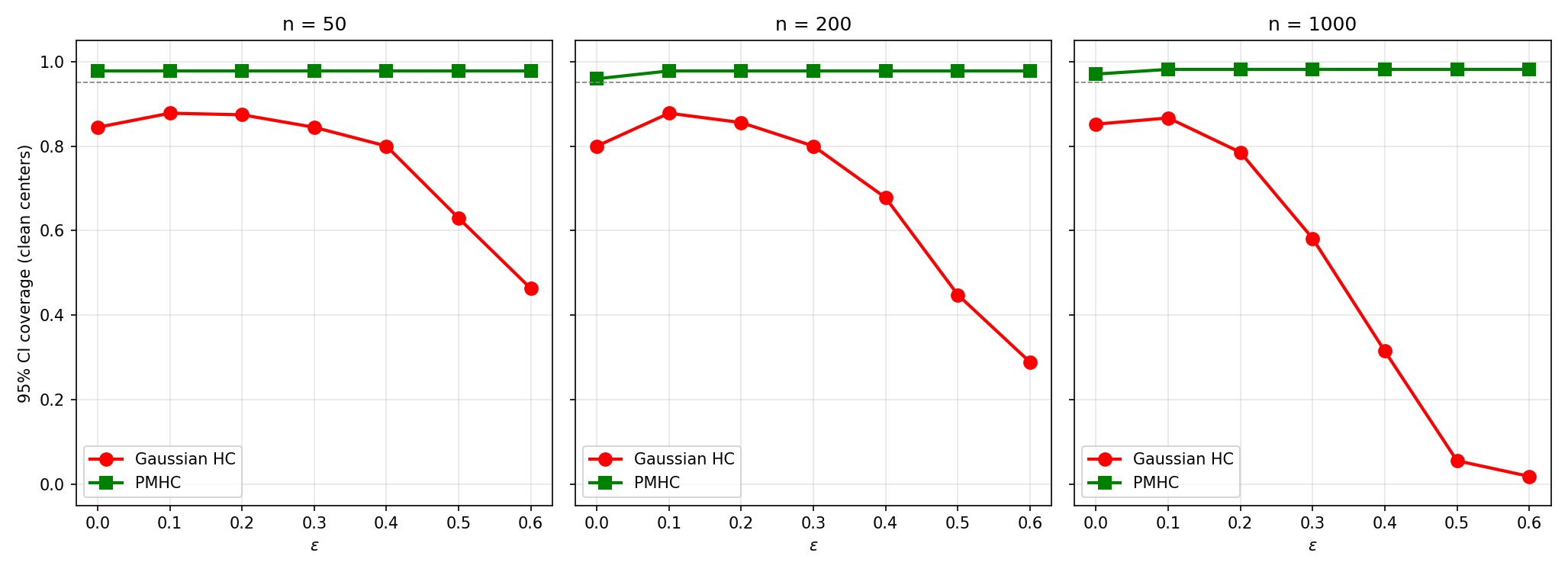}
		\caption{Coverage of the $95\%$ credible interval for clean centers as a function of the contamination fraction $\epsilon$.}
		\label{fig:coverage_vs_epsilon}
	\end{figure}
	
	\begin{figure}[H]
		\centering
		\includegraphics[width=0.95\textwidth]{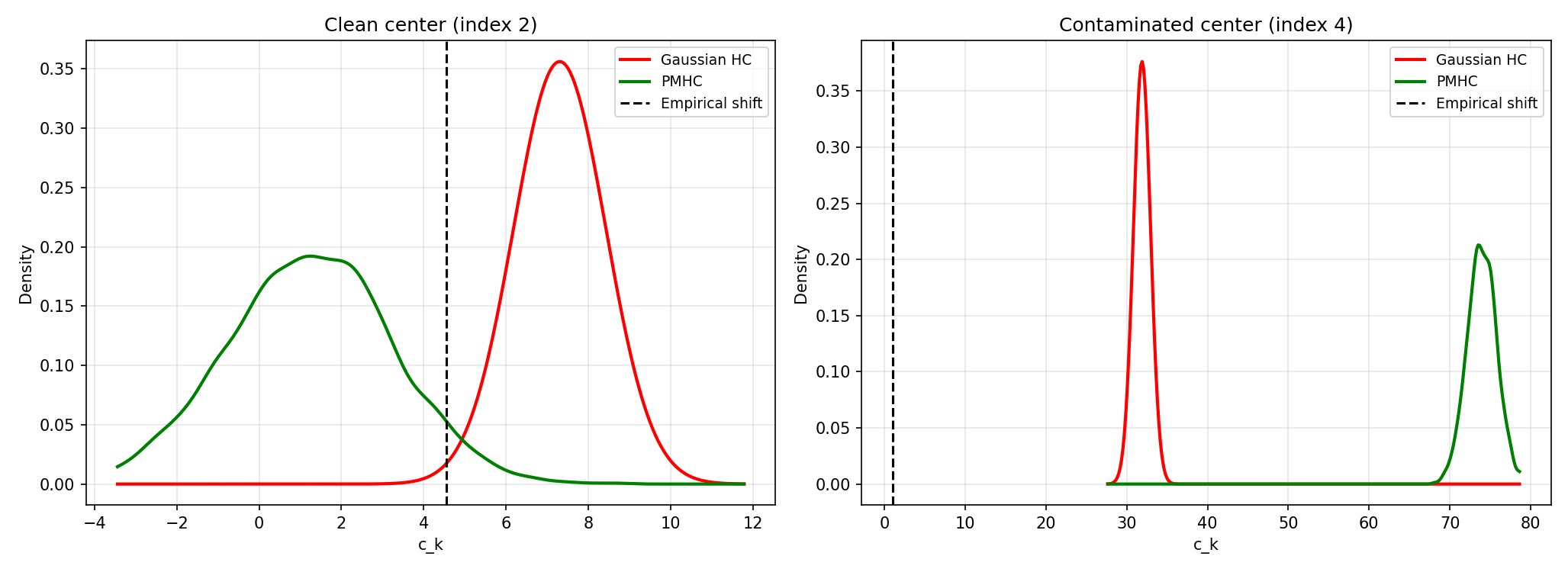}
		\caption{Posterior kernel density estimates for a clean center (index~2) and the contaminated center (index~4).}
		\label{fig:kde}
	\end{figure}
	
	\begin{figure}[H]
		\centering
		\includegraphics[width=0.98\textwidth]{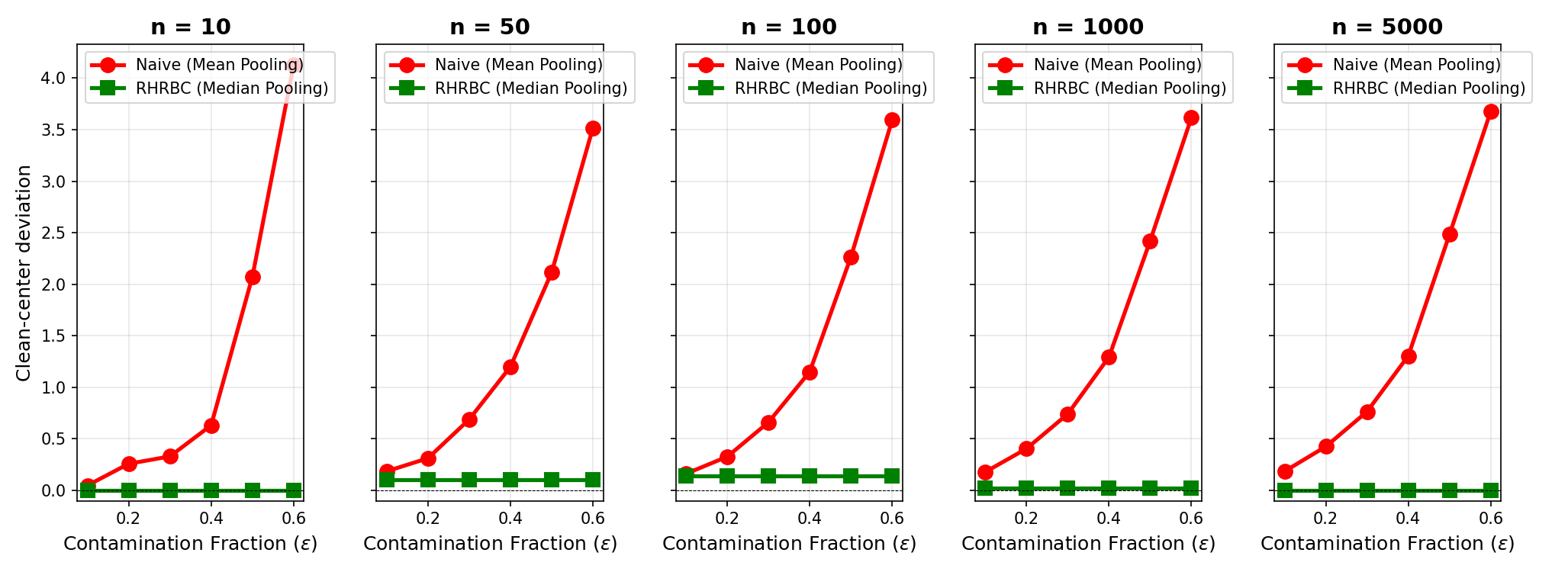}
		\caption{Finite-sample and large-sample behavior of the point estimator under contamination.}
		\label{fig:finite_sample}
	\end{figure}
	
	\begin{figure}[H]
		\centering
		\includegraphics[width=0.95\textwidth]{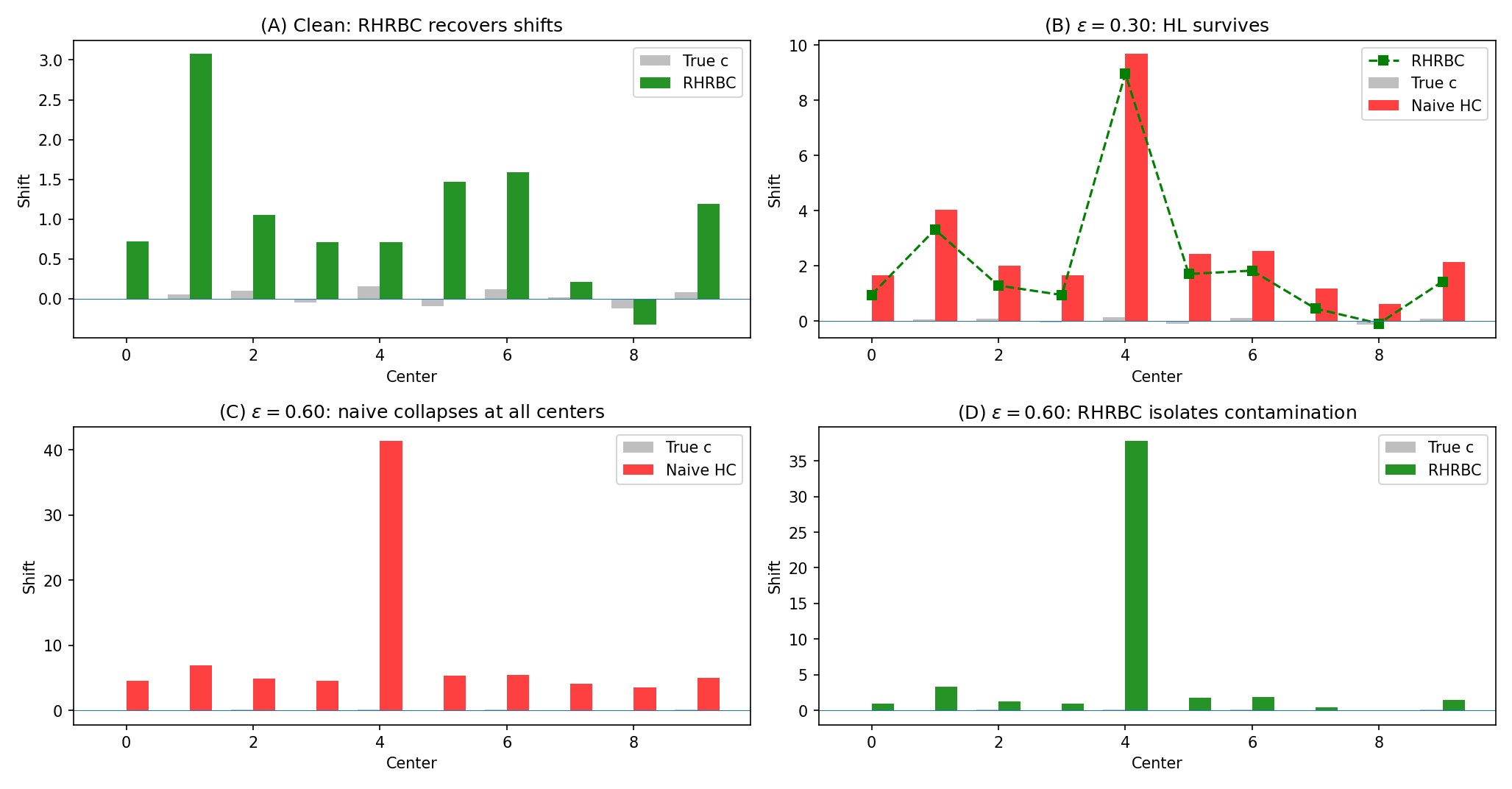}
		\caption{Per-center calibration estimates under four contamination regimes.}
		\label{fig:per_center}
	\end{figure}
	
	\begin{figure}[H]
		\centering
		\includegraphics[width=0.95\textwidth]{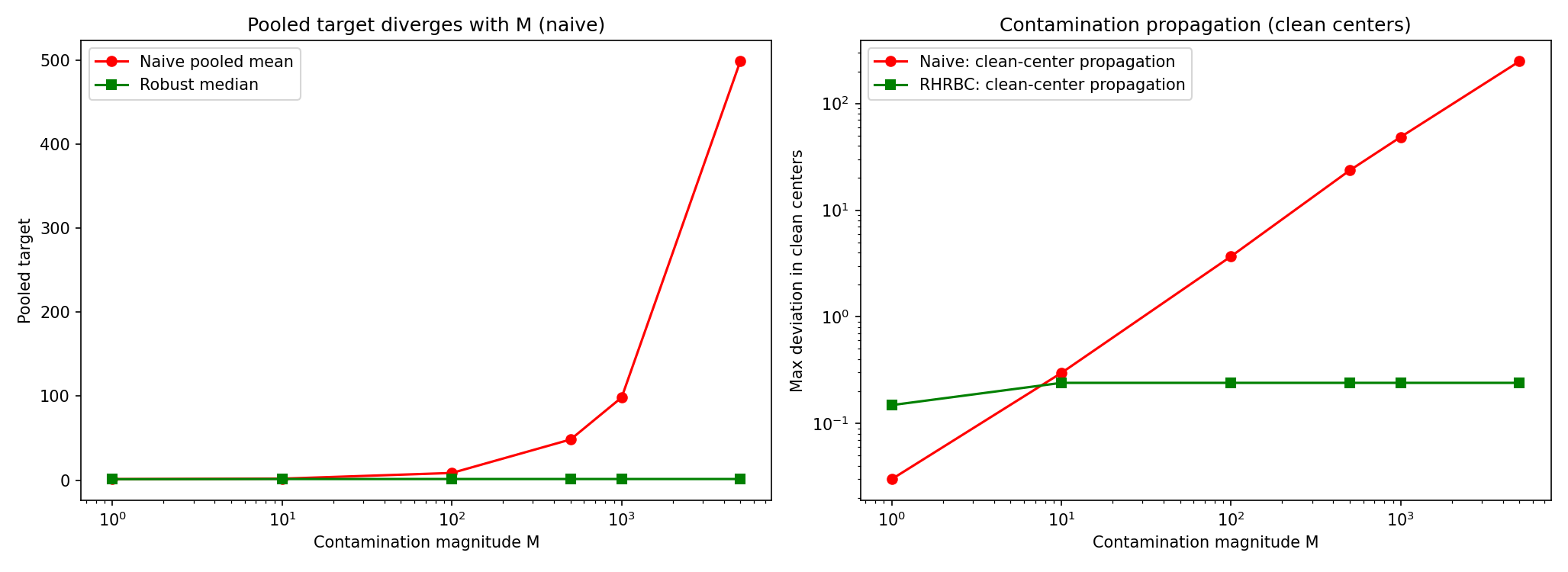}
		\caption{Pooled shrinkage target and clean-center deviation as functions of contamination magnitude $M$.}
		\label{fig:magnitude}
	\end{figure}
	
	\begin{figure}[H]
		\centering
		\includegraphics[width=0.95\textwidth]{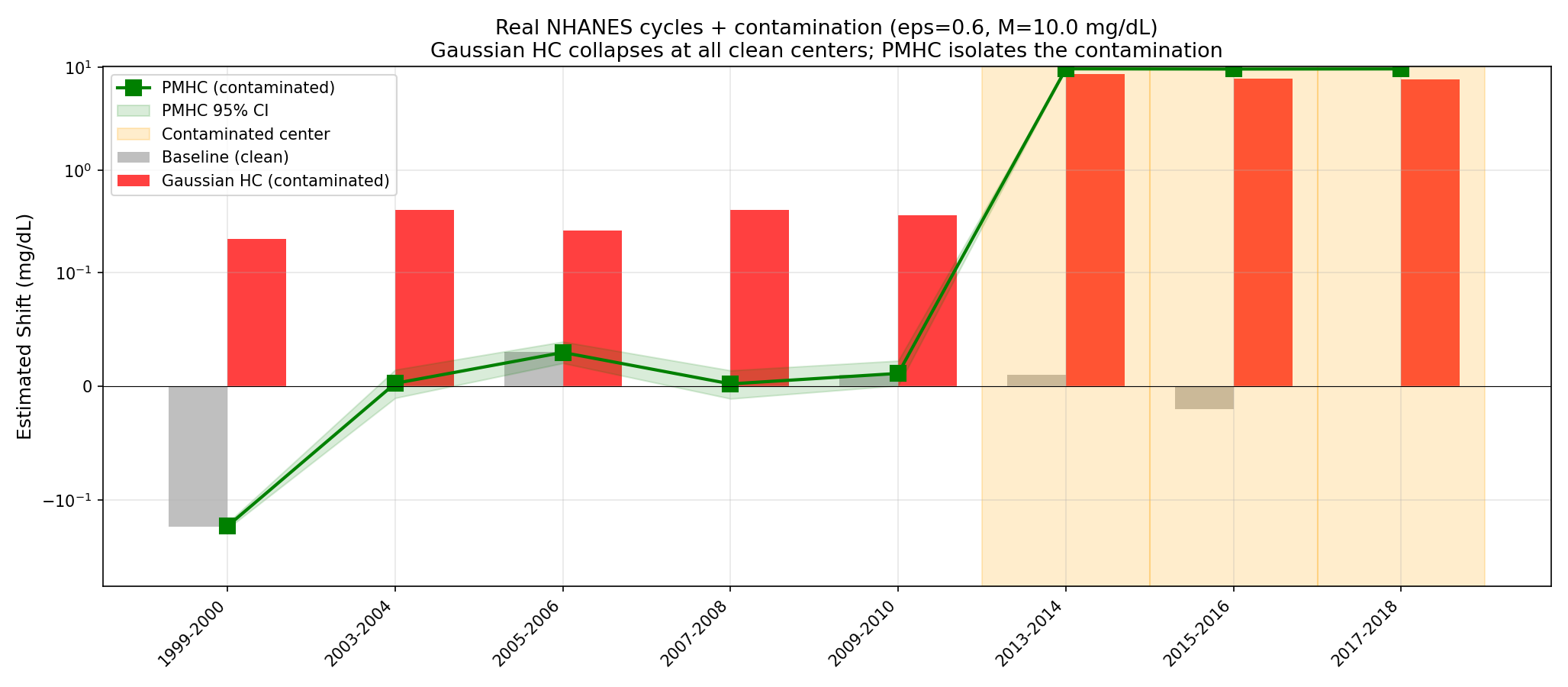}
		\caption{Real NHANES data with multi-center contamination.}
		\label{fig:nhanes_ultimate}
	\end{figure}
	
	\begin{figure}[H]
		\centering
		\includegraphics[width=0.75\textwidth]{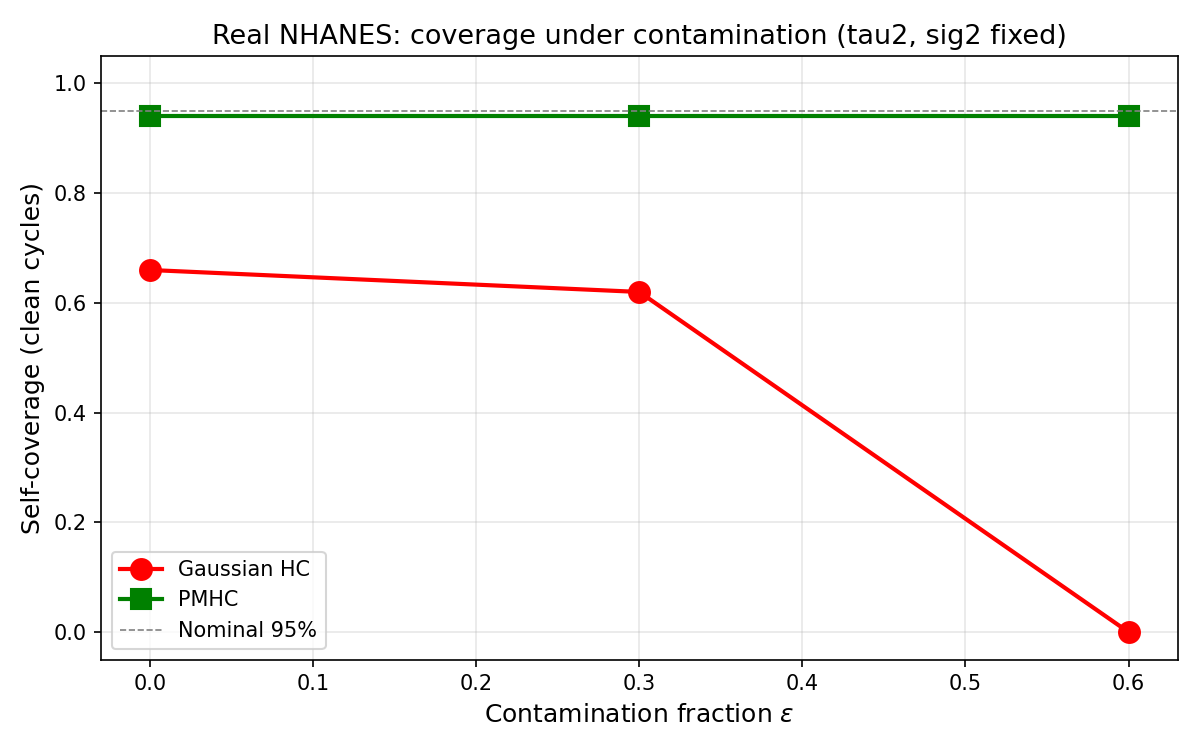}
		\caption{Self-coverage of the $95\%$ credible interval for the clean NHANES cycles.}
		\label{fig:nhanes_coverage}
	\end{figure}
	
	\begin{figure}[H]
		\centering
		\includegraphics[width=0.75\textwidth]{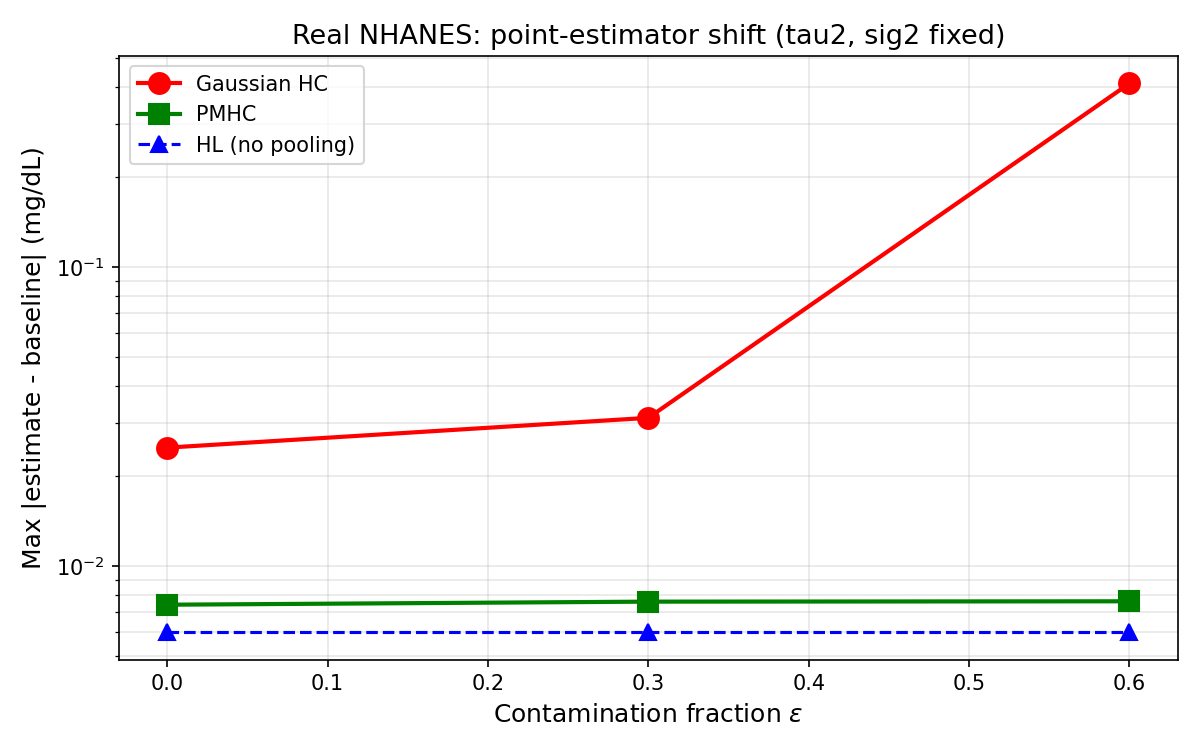}
		\caption{Point-estimator shift for the clean NHANES cycles.}
		\label{fig:nhanes_shift}
	\end{figure}
	
\end{document}